\documentclass[prx,aps,floatfix,amsmath,superscriptaddress,twocolumn,longbibliography,nofootinbib]{revtex4-2}

\usepackage{amssymb,enumerate}
\usepackage{graphicx}
\usepackage{graphics}
\usepackage{amsmath}
\usepackage{amsthm,bbm}
\usepackage{color}
\usepackage{dsfont}
\usepackage{hyperref}
\usepackage{natbib}
\usepackage{multirow}

\usepackage{wrapfig, blindtext}
\usepackage[most]{tcolorbox}
\tcbset{colback=yellow!10!white, colframe=red!50!black, 
  highlight math style= {enhanced, 
    colframe=red,colback=red!10!white,boxsep=0pt}
}
\usepackage{amsmath}
\usepackage{tikz-cd}
\usepackage{empheq}
\usepackage{graphicx}

\usepackage{anyfontsize}
\usepackage{hyperref}
\usepackage[capitalise]{cleveref}
\usepackage{nameref}
\usepackage{xpatch}
\makeatletter
\xpatchcmd{\@ssect@ltx}{\@xsect}{\protected@edef\@currentlabelname{#8}\@xsect}{}{}
\xpatchcmd{\@sect@ltx}{\@xsect}{\protected@edef\@currentlabelname{#8}\@xsect}{}{}
\makeatother
\usepackage{youngtab}

\usepackage{tikz}
\usetikzlibrary{decorations.pathreplacing}

\usepackage{tikz-cd}

\newcommand{\SU}{\operatorname{SU}}

\newcommand\End{\operatorname{L}}
\newcommand{\e}{\operatorname{e}}

\renewcommand{\i}{\mathrm{i}}

\renewcommand{\P}{\mathbf{P}}

\usepackage{makecell}
\newcommand\TopRule{\Xhline{0.08em}}

\newcommand\MidRule{\Xhline{0.03em}}
\newcommand\BotRule{\Xhline{0.08em}}

\makeatletter 
\renewcommand\onecolumngrid{%
  \do@columngrid{one}{\@ne}%
  \def\set@footnotewidth{\onecolumngrid}%
  \def\footnoterule{\kern-6pt\hrule width 1.5in\kern6pt}%
}
\makeatother 
\newlength\dlf 

\usepackage{amsmath,amssymb}
\usepackage{amssymb,enumerate}
\usepackage{graphicx}
\usepackage{graphics}
\usepackage{amsmath}
\usepackage{amsthm,bbm}
\usepackage{color}
\usepackage{dsfont}
\usepackage{bbm}

\usepackage{lipsum}
\usepackage{lmodern}
\usepackage{tcolorbox}

\usepackage{amsfonts}
\usepackage{graphicx,graphics,epsfig,times,bm,bbm,amssymb,amsmath,amsfonts,mathrsfs}
\usepackage[normalem]{ulem}
\usepackage{subfigure}
\usepackage{dsfont}
\usepackage{braket}
\usepackage{upgreek }
\usepackage{tikz}
\usepackage{natbib}
\usepackage{chngcntr}

\newtheorem{theorem}{Theorem}

\newtheorem{corollary}{Corollary}

\newtheorem{lemma}{Lemma}
\newtheorem*{lemma*}{Lemma}

\newtheorem{proposition}{Proposition}

\theoremstyle{remark}
\newtheorem{remark}{Remark}

\newcommand{\bes} {\begin{subequations}}
\newcommand{\ees} {\end{subequations}}
\newcommand{\bea} {\begin{eqnarray}}
\newcommand{\eea} {\end{eqnarray}}
\newcommand{\be} {\begin{equation}}
\newcommand{\ee} {\end{equation}}

\def\>{\rangle}
\def\<{\langle}
\def\Tr{\operatorname{Tr}}

\newcommand{\ignore}[1]{}

\usepackage{amsmath,amssymb}

\newcommand{\complex}{\mathbb{C}}

\newcommand{\identity}{\mathbb{I}}

\usepackage{mathtools}
\usepackage{microtype}

\usepackage{ytableau}
\ytableausetup{smalltableaux}

\crefname{section}{Sec.}{Secs.}
\crefname{claim}{Claim}{Claims}

\begin{document}
\title{Rotationally-Invariant  Circuits: Universality with  the exchange interaction and two ancilla qubits}
\author{Iman Marvian}
\affiliation{Departments of Physics, Duke University, Durham, NC 27708, USA}
\affiliation{Department of Electrical and Computer Engineering, Duke University, Durham, NC 27708, USA}
\author{Hanqing Liu}
\affiliation{Departments of Physics, Duke University, Durham, NC 27708, USA}
\author{Austin Hulse}
\affiliation{Departments of Physics, Duke University, Durham, NC 27708, USA}

\begin{abstract}
Universality of local unitary transformations is one of the cornerstones of quantum computing with many   applications and implications that go beyond this field. However,  it has been recently shown that this universality does not hold in the presence of continuous symmetries: generic symmetric unitaries on a composite system cannot be implemented, even approximately,   
using local symmetric unitaries on the subsystems [I. Marvian, Nature Physics (2022)].  In this work, we study qubit circuits formed from k-local rotationally-invariant  unitaries and  fully characterize the constraints imposed by locality on the realizable unitaries.
 We also present an interpretation of these  constraints in terms of the  average energy of states with a fixed angular momentum.    
Interestingly,  despite  these constraints, we show that,  using a pair of ancilla qubits, any rotationally-invariant unitary  can be realized with the Heisenberg exchange interaction, which is 2-local and rotationally-invariant.    We also show that a single ancilla is not enough to achieve universality. 
Finally, we discuss  applications of these  results for quantum computing with semiconductor quantum dots,   quantum reference frames, and resource theories.

  \end{abstract}

\maketitle

\noindent\emph{Introduction---}    Symmetric quantum circuits 
capture three basic and ubiquitous properties of quantum systems, namely unitarity, locality and symmetry. Hence, in addition to their applications in quantum computing (e.g., in   \cite{DiVincenzo:2000kx, Bacon:2000qf}) they also arise in other areas of quantum physics, from quantum thermodynamics \cite{FundLimitsNature, brandao2013resource, janzing2000thermodynamic,  lostaglio2015quantumPRX, halpern2016microcanonical, guryanova2016thermodynamics, chitambar2019quantum, lostaglio2017thermodynamic, halpern2016microcanonical, guryanova2016thermodynamics} and  reference frames \cite{QRF_BRS_07}, to  quantum chaos \cite{khemani2018operator} and classification of phases of matter  \cite{chen2010local, chen2011classification}.   
 In the absence of symmetries, it is well-known that  2-local unitaries are universal, that is,  any unitary transformation on a system of qubits  can be generated  by a sequence of unitaries acting on pairs of qubits \cite{divincenzo1995two,  lloyd1995almost, deutsch1995universality}.    
   However, perhaps surprisingly, it turns out that   this universality does not hold in the presence of symmetries: One of us has recently shown  that in the case of continuous symmetries, such as SU(2),  generic symmetric unitaries cannot be implemented, even approximately,  using local symmetric unitaries \cite{marvian2022restrictions, alhambra2022forbidden}. 

 In this Letter we study this phenomenon for qubit systems with rotational SU(2) symmetry. First, we fully characterize the constraints imposed by locality on  the realizable unitaries and explain how the violation of these constraints  can be experimentally observed. As argued in \cite{marvian2022restrictions}, this leads to a  method for probing the locality of interactions in nature.  
We also introduce a physical interpretation of these constraints in terms of the average energy of states with a fixed angular momentum.  
 
 Secondly, we find a way to circumvent  the restrictions imposed by locality: We show that any rotationally-invariant  unitary can be implemented using the Heisenberg exchange interaction, provided that the qubits in the system can interact with  a pair of ancillary qubits. It should be noted that the seminal works of Bacon and DiVincenzo, et. al  \cite{Bacon:2000qf, DiVincenzo:2000kx}
   have already established the   
 \emph{encoded} universality \cite{Bacon:Sydney, Kempe:01} of the exchange interaction for quantum computing (See also  
  \cite{kempe2001theory, bacon2001coherence, levy2002universal, rudolph2005relational}).   That is, they show that for certain encodings of quantum information in spin-half systems universal quantum computation can be achieved  using the exchange interaction  alone  (e.g., in \cite{DiVincenzo:2000kx} each qubit is encoded in 3 spins). We, on the other hand, establish a stronger notion of universality, namely we show that  \emph{all} rotationally-invariant unitaries can be implemented using the exchange interaction and a pair of ancilla qubits.

This result has direct applications for the spin qubit approach to quantum technology, where the quantum information is encoded and manipulated in spin-half  particles, such as single electrons confined to semiconductor quantum dots \cite{loss1998quantum}.  With recent advances of silicon quantum dots, voltage controllable  exchange interaction   
can be readily realized  in such systems (See, e.g., 
\cite{malinowski2019fast}).    As an example of applications, we present a scheme for implementing the family of unitaries generated by the multi-qubit swap Hamiltonian. This family has found extensive applications, e.g., for performing density matrix exponentiation \cite{lloyd2014quantum,  marvian2016universal, kimmel2017hamiltonian, pichler2016measurement, kjaergaard2022demonstration}, a key subroutine used in various quantum algorithms and protocols \cite{lloyd2014quantum, pichler2016measurement, marvian2016universal,kimmel2017hamiltonian}.  
  Although these unitaries are rotationally-invariant, 
  they cannot be implemented using local gates that respect this symmetry. We show how this can be circumvented using ancilla qubits.   Generally, implementing a desired rotationally-invariant unitary with local gates that also respect this symmetry  suppresses certain types of errors and results in better fault tolerance \cite{Bacon:2000qf, DiVincenzo:2000kx}.  
      At the end of the paper, we discuss  other applications of this result in the areas of quantum thermodynamics and quantum reference frames.\\

\noindent \emph{Preliminaries---}  Consider a system with $n$ spin-half subsystems (qubits) with the total Hilbert space $(\mathbb{C}^2)^{\otimes n}$. We say an operator $A$ on this system is rotationally-invariant, or \emph{symmetric},  if $ U^{\otimes n} A {U^\dag}^{\otimes n}=A$ for all single-qubit unitaries $U$, or, equivalently,  if  $[A, J_v]=0$ for $v= x, y, z,$  
where $J_v=\frac{1}{2}\sum_{i=1}^n \sigma^{(v)}_i$ is the angular momentum operator in $v$ direction, and $\sigma^{(v)}_i$ is a Pauli operator on qubit $i$, tensor product with the identity operators on the rest of qubits.   
 Consider quantum circuits formed from $k$-local rotationally-invariant unitaries, i.e., unitaries that act non-trivially on, \emph{at most}, $k$ qubits.  Let $\mathcal{V}_{k}$ be the group generated by composing finitely many such unitaries.  In particular, $\mathcal{V}_{n}$ is the group of all rotationally-invariant unitaries.  It can be shown that
 $\mathcal{V}_{k}$ is a compact connected Lie subgroup of $\mathcal{V}_{n}$ \cite{marvian2022restrictions}. Our 
goal here is to characterize $\mathcal{V}_{k}$, or equivalently, to determine the constraints imposed by locality on realizable unitaries. Group $\mathcal{V}_{k}$ can be equivalently defined as the family of unitaries generated by  rotationally-invariant Hamiltonians that can be written as a sum of  $k$-local terms \cite{marvian2022restrictions}.  
 Therefore,  characterizing $\mathcal{V}_{k}$ also determines general constraints on the time evolutions generated by such Hamiltonians. Special cases of interest are Hamiltonians that can be realized by tunable exchange interactions, i.e.,  $H(t)=\sum_{r<s} h_{rs}(t)\ R_{rs}$, where  $R_{rs}=({\vec{\sigma}_r\cdot \vec{\sigma}_{s}})/2=\sum_{v=x,y,z}\sigma^{(v)}_r\sigma^{(v)}_s/2$  is the exchange interaction between qubits $r$ and $s$, and $h_{rs}$ is an arbitrary real function of time.  Any unitary generated by such  Hamiltonians is in $\mathcal{V}_{2}$. Conversely, since any rotationally-invariant operator on a pair of qubits is a linear combination of the identity operator and  $\vec{\sigma}\cdot \vec{\sigma}$,   any unitary in  $\mathcal{V}_{2}$ can be realized by the exchange interaction,  up to a global phase.

To study this problem we use a Lie-algebraic approach, which has been previously used to study universality in the absence of symmetries \cite{divincenzo1995two, lloyd1995almost, brylinski2002universal, childs2010characterization,  zanardi2004universal, giorda2003universal,  Bacon:2000qf, lidar1998decoherence}. Suppose one can implement Hamiltonians $\sum_j a_j(t) A_j$,
where $\{a_j\}$ are arbitrary real functions of time and $\{A_j\}$ are Hermitian operators. Using this family of Hamiltonians one can implement unitaries $\exp({-i B t })$ for all time $t$, if and only if the skew-Hermitian operator $i B$ is in the real Lie algebra generated $\{i A_j\}$, that is, it can be written as a linear combination of  $\{i A_j\}$ and their (nested) commutators with real coefficients  \cite{d2007introduction, jurdjevic1972control}.\\

\noindent\emph{Characterizing symmetric  unitaries---} Under the action of rotations, the Hilbert space of $n$ qubits  decomposes to sectors  with different total angular momenta,  which  label inequivalent irreducible representations (irreps) of SU(2).  The total squared angular momentum $J^2=J^2_x+J^2_y+J_z^2$, also known as the Casimir operator, has eigenvalues $j(j+1)$ with $j=j_{\min}, j_{\min}+1,\cdots,  j_{\max}$, where $j_{\max}=n/2$, and $j_{\min}=0,1/2$,
 for even and odd $n$, respectively.  
As reviewed in  Appendix \ref{App:A},  angular momentum $j$ 
corresponds to an irrep with dimension $2j+1$, which appears with the multiplicity \be\label{eq:hook_length}
m(n,j)={\binom{n}{\frac{n}{2}-j}}\times \frac{2j+1}{\frac{n}{2}+j+1}\ .
\ee
Then,  the total Hilbert space decomposes as 
$(\mathbb{C}^2)^{\otimes n} \cong \bigoplus_{j=j_\text{min}}^{j_\text{max}} \mathcal{H}_j\cong \bigoplus_{j=j_\text{min}}^{j_\text{max}} \mathbb{C}^{2j+1} \otimes  \mathbb{C}^{m(n,j)} $. Here, $\mathcal{H}_j$ is the eigen-subspace of $J^2$ with eigenvalue $j(j+1)$, also known as the subspace of states with angular  momentum $j$,  and $\mathbb{C}^{m(n,j)}$ is the  multiplicity subsystem, where $\SU(2)$ acts trivially \cite{QRF_BRS_07, zanardi2001virtual}.     
Using this decomposition together with  Schur's  lemmas one can characterize  rotationally-invariant unitaries \cite{QRF_BRS_07}: they are  block-diagonal with respect to $\{\mathcal{H}_j\}$  and act trivially on the irreps of SU(2). That is,  unitary $V\in\mathcal{V}_n$ if, and only if, it can be decomposed  as $V\cong\bigoplus_j (\mathbb{I}_{2j+1} \otimes v_j)$, where $ \mathbb{I}_{2j+1}$ is the identity operator on $(2j+1)$-dimensional irrep of SU(2),  and $v_j$ is an arbitrary unitary on  $\mathbb{C}^{m(n,j)}$.   \\

\noindent\emph{Constraints imposed by locality---}  
 In the case of qubit systems with SU(2) symmetry the restrictions imposed by locality  are limited to constraints on the relative phases between the subspaces with different irreps of symmetry  (Interestingly, in the case of qudits with SU($d$) symmetry for $d\ge 3$, there are stronger constraints \cite{Marvian2021qudit}).  In terms of realizable Hamiltonians, as stated in theorem \ref{newThm} below, this amounts to constraints on the inner products of the Hamiltonian with the projectors to subspaces $\{\mathcal{H}_j\}$, denoted by  $\{\Pi_j: j=j_\text{min},\cdots, j_\text{max}\}$. 
The space spanned by these projectors, denoted by $\mathcal{C}$,  is  the space of operators that are invariant  under all rotations and permutations, which has dimension $\lfloor n/2\rfloor+1$ (See Appendix \ref{App}).   To express the constraints imposed by locality it is useful to introduce another basis for $\mathcal{C}$, namely operators $\{C_l\}$  labeled by  even integers $l$, where $C_0$ is the identity operator and for $l=2, 4,\cdots,2\lfloor{{n}/{2}}\rfloor$, 
\be\label{defCl}
C_l\equiv \frac{1}{(l/2)!} \hspace{1mm}\sum_{i_1\neq \cdots \neq i_{l}} \hspace{-3mm} R_{i_1,i_2}\cdots\cdot R_{i_{l-1},i_{l}}  =\sum_{j=j_\text{min}}^{j_\text{max}} c_{l}(j) \ \Pi_j\ ,
\ee
where in the first summation $i_1, i_2, \cdots, i_{l}$ are $l$ distinct integers between 1 to $n$.  In the second summation 
\be\label{clj1}
c_{l}(j) =  \frac{l!}{2^{l/2}(l/2)!}{{n}\choose{l}}  \sum_{r=0}^{l/2} (-4)^{r}  {{l/2}\choose{r}}   \frac{m(n-2r,j)}{m(n,j)}
 \ee 
is the eigenvalue of $C_l$ in the subspace $\mathcal{H}_j$ (we use the convention  $m(a,b)=0$ for $b>a/2$).   
These Hermitian operators satisfy the following crucial properties, which are shown in  Appendix \ref{App}: First, although it is not clear from the above formula, $c_{l}(j)$  is an integer-valued polynomial of degree $ l/2$ of $j(j+1)$. This property becomes  relevant later in Eqs.(\ref{avg},\ref{ary4}).   For instance, for $l=2$ we obtain $c_{2}(j)=2j(j+1)-3n/2$ and $C_2=2\sum_{i_1<i_2} R_{i_1, i_2}$, which     is the Hamiltonian with equal exchange interactions between all pairs of qubits.   Table \ref{default} shows integers  $\{c_l(j)\}$ for $n=10$ qubits.  The second  property, which explains the  labeling of these operators with integers $l=0, 2, \cdots $,  is  that $C_l$ is a sum of $l$-local operators and is orthogonal to $k$-local operators with $k<l$.  
Indeed, Hermitian operators $\{C_l\}$ form an orthogonal basis for $\mathcal{C}$, i.e., $\Tr(C_lC_{l'})=\delta_{l,l'} \Tr(C_l^2)$, which means  
  $\sum_j \Tr(\Pi_j) c_l(j) c_{l'}(j)=\delta_{l,l'}\Tr(C^2_l)$.  Using this basis, the restrictions imposed by locality find  a simple form.

 \begin{theorem}\label{newThm}
For a system with $n$ qubits, the family of unitary evolutions $\exp({-i H t}): t\in\mathbb{R}$ generated by a rotationally-invariant Hamiltonian $H$   
 can be implemented using $k$-local rotationally-invariant unitaries with $k\ge 2$, 
if, and only if,  for all even integers $l=2\lfloor{{k}/{2}}\rfloor+2,\cdots,  2\lfloor{{n}/{2}}\rfloor$, 
it holds that
\begin{align}\label{const3}
\Tr(H C_l)&=\sum_{j=j_\text{min}}^{j_\text{max}} c_l(j)\Tr(H\Pi_j)=0\ .
\end{align}
\end{theorem}
Hence, the constraints imposed by locality can be  expressed in terms of the vector defined by  $\Tr(H\Pi_j): j=j_\text{min}, \cdots, j_\text{max}$, which in \cite{marvian2022restrictions} is called 
the \emph{charge vector} of $H$. This theorem in particular implies that  unitaries generated by Hamiltonian 
$H$ can be realized (up to a global phase) via the exchange interaction if,  and only if, Eq.(\ref{const3}) holds for even integers $l=4,\cdots  , 2\lfloor{{n}/{2}}\rfloor$. For instance, for a system with $n=4$ qubits,  this amounts to the condition $15 \Tr(H\Pi_{0})-5 \Tr(H\Pi_{1})+3 \Tr(H\Pi_{2})=0$, which, in particular, excludes $H=R_{12}R_{34}$ (See Appendix \ref{App}).  
 \begin{table}[htp]
\begin{center}
\begin{tabular}{|c|c|c|c|c|c|c|}
\cline{2-7}
\multicolumn{1}{c|}{}& \multicolumn{6}{|c|}{Angular Momentum} 
\\ \hline
   &$j=0$\ & \ $j=1$\ &\  $j=2$\ &\  $j=3$\ &\  $j=4$\ &\  $j=5$\ \  \\
\hline
$l=0$ \text{body}  &1 & 1 & 1& 1& 1& 1\\
\hline
$l=2$ \text{body}  &-15 & -11 & -3& 9& 25& 45\\
\hline
$l=4$ \text{body}   &150 & 70 & -42& -90& 70& 630\\
\hline
$l=6$ \text{body}   &-1050 & -210 & 462 & -90 & -1050 & 3150\\ \hline
$l=8$ \text{body}&4725 & -315& -1323& 2565& -3675 & 4725\\ \hline
$l=10$ \text{body}&-10395 &  3465 & -2079& 1485 & -1155& 945\\
\hline
\end{tabular}
\end{center}
\caption{Integers $\{c_l(j)\}$, which are the eigenvalues of operators $\{C_l\}$, for   a system with $n=10$ qubits. Up to a normalization,  $c_l(j)$ is the average energy of states with angular momentum $j$, under $l$-body interactions (See Eq.\ref{avg}). $c_l(j)$ is an integer-valued polynomial of degree $l/2$ of $j(j+1)$.  }
\label{default}
\end{table}%
 
 As we explain in Appendix \ref{App:proofThem}, the necessity of these conditions follows from the fact that for the above values of $l$,  $\{C_l\}$ are orthogonal to $k$-local operators and commute with rotationally-invariant operators.    To show the converse statement, first note that any Hamiltonian $H$ decomposes as    
\be\label{dec-22}
H=H_0+\sum_j \frac{\Tr(H \Pi_j)}{\Tr(\Pi_j)}\ \Pi_j=H_0+\sum_l \frac{\Tr(H C_l)}{\Tr(C^2_l)}\ C_l\ ,
\ee
where $H_0$ is orthogonal to $\mathcal{C}$, i.e.,  $\Tr(H_0\Pi_j)=\Tr(H_0 C_l)=0$, for all $l=0, 2, \cdots , 2\lfloor{{n}/{2}}\rfloor$ and all $j=j_\text{min},\cdots, j_\text{max}$.  Using  rather elementary  techniques, in Appendix \ref{App:proofThem}   we show that any Hamiltonian $H_0$ satisfying these constraints can be realized using the exchange  interaction,  which is 2-local and rotationally-invariant (We note that this  can also be shown using  more advanced results in the mathematical literature, namely the result of Marin that finds 
      the decomposition into simple factors of the Lie algebra generated by transpositions  \cite{marin2007algebre}). Furthermore,  since for $l\le k$  operator $C_l$ can be written as a sum of $k$-local symmetric Hermitian operators,   Hamiltonian $C_l$ can be realized by such Hamiltonians. Combining these facts we find  that any Hamiltonian satisfying the conditions in theorem \ref{newThm} can be implemented using $k$-local  symmetric Hamiltonians. Note that this result implies that 
for any symmetric unitary $V$, there exists $\theta_j\in[-\pi,\pi)$, such that the unitary $V[\sum_j \e^{\i \theta_j} \Pi_j]$ can be realized using the exchange interaction.

Eq.(\ref{const3}) imposes $ \lfloor{{n}/{2}}\rfloor-\lfloor{{k}/{2}}\rfloor$ independent constraints on the manifold of realizable Hamiltonians.    
 It follows that the difference between the dimensions of the group of all symmetric unitaries and the subgroup generated by $k$-local  symmetric unitaries is $\lfloor{{n}/{2}}\rfloor- \lfloor{{k}/{2}}\rfloor$,  matching  the general lower bound  in \cite{marvian2022restrictions}. In fact, in Appendix \ref{App:dim} we show that  for $k\ge 2$, 
  \be\label{dim}
  \dim(\mathcal{V}_{k}) =\frac{1}{n+1}\binom{2n}{n} -\lfloor{\frac{n}{2}}\rfloor+ \lfloor{\frac{k}{2}}\rfloor\ \ ,
  \ee
  where $\frac{1}{n+1}\binom{2n}{n}$ is the $n$th Catalan number.     Thus, unless $ \lfloor{{k}/{2}}\rfloor=\lfloor{{n}/{2}}\rfloor$,    $k$-local symmetric unitaries are not universal.   Interestingly,   universality can be achieved with $(n-1)$-local symmetric unitaries if $n$ (the number of qubits) is odd, whereas such unitaries are not universal for even $n$. In Appendix \ref{App:TimeReversal} we discuss more about this even/odd effect and its connection with the time-reversal symmetry.  Finally, note that since  $\mathcal{V}_k$ is compact \cite{marvian2022restrictions},  if a symmetric unitary $V$  does not belong to $\mathcal{V}_{k}$, then there is a neighborhood of unitaries  around $V$, none of which  can be implemented using $k$-local symmetric  unitaries.   \\

  In the rest of the paper we present three independent results about the constraints in Eq.(\ref{const3}): First,
we discuss a physical interpretation of them, then we explain  how their violations can be experimentally observed, and finally we show how they can be circumvented using ancillary qubits.\\

\noindent\emph{Average energy for a fixed angular momentum---} To understand the  constraints in Eq.(\ref{const3}) better, we consider the average energy of states with angular momentum $j$, i.e., the expectation value of Hamiltonian for the maximally-mixed state over $\mathcal{H}_j$. In Appendix  \ref{App:avg} we show that if a general (possibly symmetry-breaking) Hamiltonian $H$ can be written as a sum of $k$-local terms then this average energy is 
 \be\label{avg}
 E_j\equiv \frac{\Tr(\Pi_j H)}{\Tr(\Pi_j)}=\sum_{l=0}^{2\lfloor{{k}/{2}}\rfloor} q_l \times c_l(j) \ ,
 \ee 
where  $q_l=\Tr(H C_l)/\Tr(C_l^2)$, and the summation is over even integers.   The right-hand side is a polynomial of degree $\lfloor{{k}/{2}}\rfloor$ of $j(j+1)$, and it can be an arbitrary function of $j=j_\text{min}, \cdots, n/2$ if, and only if, $\lfloor{{k}/{2}}\rfloor=\lfloor{{n}/{2}}\rfloor$.   For instance, for $k=2$, $E_j=q_0+ q_2[2j(j+1)-{3n}/{2}]$. Therefore,  in general (regardless of symmetry) the locality of interactions imposes strong constraints on the form of dependence of the average energies $\{E_j\}$ to the  angular momentum $j$ (it should be quadratic in this example). 
 Now,  in the presence of rotational symmetry, due to the conservation of angular momentum,  these constraints leave an observable effect  on the realized unitaries (see the next section). Furthermore, the constraints hold more generally,  including for non-local Hamiltonians that can be realized by local ones:  for Hamiltonian $H$, if for all $t\in\mathbb{R}$,  $\exp(-i H t)\in \mathcal{V}_k$ (which means 
it can be implemented by $k$-local symmetric  unitaries) then the average energies $\{E_j\}$ satisfy Eq.(\ref{avg}) for some  $\{q_l\in\mathbb{R}\}$. This can be shown directly, using the properties of  operators $\{C_l\}$, or via Eq.(\ref{const3}) together with the orthogonality relation $\sum_j \Tr(\Pi_j) c_l(j) c_{l'}(j)=\delta_{l,l'} \Tr(C_l^2)$.  Conversely, applying this relation together with  Eq.(\ref{avg}) one obtains Eq.(\ref{const3}).  We present further details in Appendix \ref{App:avg}.  Incidentally, the fact that the average energies in \cref{avg}, and the constraints in \cref{const3}, depend only on the $l$-body properties of $H$ for even $l$ can be understood in terms of the time-reversal symmetry of the squared angular momentum $J^2$ (see \cref{App:TimeReversal}).\\

\noindent \emph{$l$-body phases--}  In \cite{marvian2022restrictions} one of us introduced the idea of $l$-body phases for U(1)-invariant unitaries, which express the  constraints  imposed by locality in terms of experimentally observable quantities. Using the properties of operators  $\{C_l\}$,  this idea can be extended to  SU(2). Consider an $n$-qubit symmetric unitary $V$, which can be the unitary generated by a symmetric Hamiltonian $H(t)$, from time $t=0$ to $T$ under the Schr\"{o}dinger equation.   Any such unitary decomposes as $V=\bigoplus_j V_j$, where
$V_j$ is the component of $V$ inside $\mathcal{H}_j$, the subspace with angular momentum $j$. For even integers $l=0,\cdots , 2\lfloor{{n}/{2}}\rfloor$, define   the $l$-body phase $\Phi_l\in(-\pi,\pi]$  of $V$ as  
\be\label{ary4}
 \Phi_l\equiv    \sum_{j=j_\text{min}}^{j_\text{max}} c_l(j)\ \theta_j=-\int_0^T dt\   \Tr(H(t) C_l)\ \ \ \text{: mod} \ 2\pi\ ,
\ee
where  $\theta_j=\text{arg}(\text{det}(V_j))$ is the phase of the determinant of $V_j$.   Note that while $\theta_j$ is only defined modulo $2\pi$, because coefficients $c_l(j)$ are integer, 
$\Phi_l\in(-\pi,\pi]$ is well-defined.  As we show in Appendix \ref{App:l-body},  the above  equality is satisfied for any symmetric Hamiltonian $H(t)$ that realizes unitary $V$.  The notion of $l$-body phases provides a useful characterization  of  the constraints imposed by the locality of interactions. 
Following \cite{marvian2022restrictions},  in Appendix \ref{App:l-body} we show that: $\textbf{(i)}$ for $l\ge 2$, all $l$-body phases $\{\Phi_l\}$ can be measured experimentally, whereas the phases $\{\theta_j\}$ are not  observable, because they transform non-trivially under $V\rightarrow e^{i\alpha} V$.   Similarly,  $\Phi_0=\sum_j\theta_j=\text{arg}(\text{det}(V))$ is not  observable.  $\textbf{(ii)}$ If $V$ is realizable by $k$-local symmetric unitaries then $\Phi_l=0$ for $l> k$.    
$\textbf{(iii)}$  Conversely, for a general symmetric unitary  $V$, if all $l$-body phases vanish for $l>k$, then  $V$ is realizable using $k$-local symmetric unitaries, up to a unitary in a fixed finite subgroup of  symmetric  unitaries. Finally, as mentioned in \cite{marvian2022restrictions},   from a geometrical point of view, the transformation $\{\theta_j\} \rightarrow \{\Phi_l\}$ in Eq.(\ref{ary4})  describes  a change of the coordinate system on the  $(\lfloor{{n}/{2}}\rfloor+1)$-torus corresponding to phases $\theta_j=\text{arg}(\text{det}(V_j))$, for  $j=j_\text{min},\cdots, j_\text{max} $.  

The fact that $l$-body phases $\Phi_l$ are physically observable  for $l\ge 2$, hints to an interesting implication of our results: by measuring these phases it is possible to detect the locality of interactions; $\Phi_l\neq 0$ indicates the presence  an interaction that couples, at least, $l$ spin-half systems together. \\

\noindent\emph{Universality with ancilla qubits--}  
Next, we show how the constraints imposed by locality can be circumvented without breaking the symmetry.  Let
 $F$ be an arbitrary 
symmetric Hamiltonian that acts trivially on a pair of qubits  in the system, $r$ and $s$. 
 Consider ancilla qubits $a$ and $b$ and define Hamiltonian
\bes\label{Eq:F}
\begin{align}
\widetilde{F}&\equiv F R_{rs}\otimes I_{ab}-F\otimes R_{ab}\\ &=(F R_{rs}- \frac{F}{2})\otimes P_{ab}^{+}+(F R_{rs}+\frac{3F}{2} )\otimes P_{ab}^{-}\ ,
\end{align}
\ees
where $I_{ab}$ is the identity operator on the 4D Hilbert space of the ancillae, and  $P_{ab}^{+}$ and $P_{ab}^{-}$ are, respectively,  the projectors to their triplet and singlet (i.e., the symmetric and totally anti-symmetric) subspaces.  Using theorem \ref{newThm}, we find that symmetric Hamiltonian $\widetilde{F}$ can be implemented using the  exchange interaction: the  terms  
$F R_{rs}\otimes I_{ab}$ and $F \otimes R_{ab}$ are equal  up to a permutation that exchanges qubits $r s$ and $a b$.  Since the total angular momentum remains conserved under permutations, the inner products of the projector to any angular momentum sector with these two terms are equal. Hence Hamiltonian $\widetilde{F}$  satisfies 
the condition in theorem \ref{newThm} for all $l\ge 0$,  and therefore can be implemented using the exchange  interaction. 

Now suppose both ancilla qubits are prepared in state $|0\rangle$, and are coupled  to the $n$-qubit system in arbitrary initial state $|\psi\rangle$ via Hamiltonian $\widetilde{F}$. After time $t$, the joint state evolves to 
\be
\e^{-i \widetilde{F} t} (|\psi\rangle\otimes |00\rangle_{ab})=\big(\e^{-i [F R_{rs}- F/2] t}|\psi\rangle\big)\otimes |00\rangle_{ab}\ .
\ee
Therefore, at the end of the process the ancillae return to their initial state, whereas 
the system evolves according to the  Hamiltonian $F R_{rs}- F/2$. Combining this with Hamiltonian $F/2$, one can implement Hamiltonian $F R_{rs}$.  For example, using this procedure starting from the Hamiltonian $F=R_{12}$, one obtains Hamiltonian $R_{12} R_{34}$ using only the exchange interaction, whereas this is impossible without ancillae. 
 Repeating this procedure recursively, one obtains $R_{1,2}R_{3,4}\cdots R_{l-1,l}$ and hence Hamiltonians $C_l$ for all $l\ge 2$. Finally, applying Eq.(\ref{dec-22}),  any symmetric Hamiltonian $H$ decomposes as a term $H_0$ that is realizable via the exchange  interaction plus a linear combination of $\{C_l\}$, which can be realized using this procedure.   
 
 In conclusion,  any rotationally-invariant Hamiltonian is realizable using the exchange  interaction, provided that the system interacts with a pair of ancillae in state $|00\rangle$. Furthermore, rotational symmetry implies that, rather than this state, ancillae can be prepared in any state with support restricted to the triplet subspace, such as the maximally-mixed state over this subspace.  Alternatively, they can be prepared in the singlet $(|01\rangle-|01\rangle)/\sqrt{2}$, in which case,  under   $\widetilde{F}$  the main system evolves according to the Hamiltonian $FR_{rs}+3 F/2$.  Then,  following a similar construction,  all symmetric Hamiltonians can be obtained from the  exchange  interaction. In conclusion, to circumvent the constraints imposed by locality the state of ancilla qubits does not need to break the symmetry.   It is worth noting that this technique can be 
generalized beyond SU(2),  to achieve universality anytime  the restrictions on realizable symmetric unitaries are limited to constraints on the relative phases between sectors with inequivalent irreps of symmetry.

In Appendix \ref{App:univ}, we present  this argument more formally. Moreover, for a general class of Hamiltonians, we explicitly construct the corresponding  Hamiltonian that 
couples the system to  ancillae, in terms of  
  exchange interactions $\{R_{rs}\}$ and their nested commutators (see the example below). Finally, in Appendix \ref{App:single}, using the notion of $l$-body phases, we prove that universality cannot be achieved with a single ancilla qubit.  

\vspace{3mm}
  
\noindent \emph{Multi-qubit swap Hamiltonian--} Consider the family of unitaries generated by the multi-qubit swap Hamiltonian:   
For a system with $n=2r$ qubits partitioned to two subsystems $A$ and $B$ each with $r$ qubits, let  $S_{\text{AB}}$ be the  swap operator that exchanges their states. The family of unitaries $\exp(i \phi S_{\text{AB}})$ for $\phi\in[0,2\pi)$ appears 
 as a subroutine  in various quantum protocols (See, e.g.,\cite{lloyd2014quantum,  marvian2016universal,  kimmel2017hamiltonian,  pichler2016measurement}). 
 Applying Theorem \ref{newThm}, we find that, despite its rotational symmetry,   
for generic values of $\phi$ the unitary $\exp(i \phi S_{\text{AB}})$ is not realizable with $k$-local symmetric unitaries with $k<n$   
 (See Appendix \ref{App:Example}). 
For instance, in the case of $n=4$, the swap operator that exchanges 
qubit 1 with 2  and qubit 3 with 4 is  
$S_{AB}=\mathbb{I}/4+(R_{12}+R_{34})/2+R_{12}R_{34}$. Since $\Tr(S_{AB}C_4)>0$  this Hamiltonian cannot  be realized by 3-local symmetric Hamiltonians.  On the other hand, using a pair of ancilla qubits  this Hamiltonian can be realized via the exchange interaction. This can be shown using the above general argument, or more explicitly, using the interesting identity 
 \begin{align}
 4(R_{12} R_{34} - R_{34} R_{56}) = &-\big[ [[[R_{12}, R_{23}], R_{34}], R_{45}], R_{51}\big] \nonumber\\ &-\big[ [[[R_{23}, R_{34}], R_{45}], R_{56}], R_{62}\big]\nonumber\\  &\hspace{-12mm}-\small{\big[ [R_{13}, R_{32}], R_{24} \big] 
  +\big[ [R_{35}, R_{54}], R_{46} \big]} \ .\nonumber 
  \end{align}  
Assuming qubits 5 and 6 are ancillae prepared  in state $|00\rangle$,  by applying the above Hamiltonian one realizes Hamiltonian  
$R_{12} R_{34} - R_{34}/2$ on qubits 1-4. Finally,  combining this with  exchange interactions $R_{34}$ and $R_{12}$, one obtains  Hamiltonian $S_{AB}$, up to a constant shift $\mathbb{I}/4$. 
See Appendix \ref{App:exp} for further discussion and generalization of this technique.\\

\noindent\emph{Discussion}-- Besides quantum computing, this result has also direct  applications in the related areas of  quantum reference frames \cite{QRF_BRS_07, marvian2008building, yang2017units}, covariant quantum error correcting codes \cite{faist2020continuous, hayden2021error, yang2020covariant, kong2021near},  
  the resource theory of asymmetry \cite{gour2008resource, Marvian_thesis, marvian2013theory},   and the resource theory of quantum thermodynamics    
with SU(2) conserved charges \cite{lostaglio2017thermodynamic, halpern2016microcanonical, guryanova2016thermodynamics}. In all these areas one is often interested in implementing a general symmetric  unitary. For instance,  it is often assumed that, without access to the standard Cartesian reference frame relative to which states, operations  and observables are defined, 
it is still possible to implement a general   rotationally-invariant unitary on a system with an arbitrarily large number of qubits \cite{QRF_BRS_07}.  Similarly,  the resource theory of asymmetry for SU(2) symmetry is defined based on the assumption 
that such unitaries can be implemented  with negligible costs (A similar assumption is made in quantum thermodynamics with SU(2) conserved charges). However, the no-go theorem found in \cite{marvian2022restrictions} raises a question about the justifiability of these assumptions. It suggests that to implement a general symmetric unitary 
with local operations, one may need to break the symmetry.  
Ref. \cite{marvian2022restrictions} 
 shows that in the case of
U(1) symmetry, this no-go theorem can be circumvented using an ancilla qubit.  
The present paper extended this to  the case of SU(2) symmetry with spin-half systems. These results   vindicate 
 the fundamental assumptions of the resource theory of asymmetry and  reference frames, at least, in the special cases of U(1) and SU(2) symmetry.   \\
 
\noindent\emph{Acknowledgments}--This work was supported by  NSF Phy-2046195, NSF QLCI grant OMA-2120757 and ARL-ARO QCISS grant number 313-1049. HL is supported by the U.S. Department of Energy, Office of Science, Nuclear Physics program under Award Number DE-FG02-05ER41368.

\bibliography{Ref_2021_v3, Ref_2020}

\newpage

\newpage

\onecolumngrid

\newpage

\maketitle
\vspace{-5in}
\begin{center}

\Large{Supplementary Material: Rotationally-Invariant Quantum Circuits  }

$$.  $$
\end{center}
\appendix

\newcommand\appitem[2]{\hyperref[{#1}]{\textbf{\cref*{#1}: \nameref*{#1}}} \dotfill \pageref{#1}\\ \begin{minipage}[t]{0.8\textwidth} #2\end{minipage}}\begin{itemize}
\item \appitem{App:A}{We briefly review the decomposition of qubit systems into irreps with fixed total angular momenta, as well as Schur-Weyl duality, which characterizes rotationally-and permutationally-invariant operators.}
\item \appitem{App}{We prove properties of the $C_l$ operators and determine their eigenvalues $c_l(j)$ as polynomials of $j (j + 1)$.}
\item \appitem{App:proofThem}{We prove \cref{newThm}, which completely characterizes the Hamiltonians realizable by $k$-local rotationally-invariant interactions.}
\item \appitem{App:dim}{Using Lie-algebraic techniques, we determine the dimension of the Lie group $\mathcal{V}_k$.}
\item \appitem{App:avg}{Using the basis $\set{C_l}$, we determine the average energy in a fixed angular momentum sector for Hamiltonian $H$, and we discuss further consequences when $H$ can be generated by $k$-local rotationally-invariant interactions.}
\item \appitem{App:TimeReversal}{We discuss the restrictions on $k$-local rotationally-invariant interactions in relation to time-reversal symmetry.}

\item \appitem{App:l-body}{It is shown that $l$-body phases are global-phase independent (and hence measurable), and we show that, up to a finite group, any rotationally invariant unitary can be written in terms of its $l$-body phases and a unitary which can be realized by exchange interactions.}

\item \appitem{App:Example}{Using \cref{newThm}, we show that a Hamiltonian which swaps the two parts of a system decomposed into halves cannot be realized with $k$-local rotationally-invariant interactions when $k < n$.}

\item \appitem{App:univ}{We prove that two ancilla qubits enable the implementation of arbitrary rotationally-invariant unitaries using the exchange interaction. We also provide an explicit construction for a certain class of Hamiltonians.}
\item \appitem{App:single}{It is shown that a single ancilla qubit cannot be used to achieve universality for rotationally-invariant unitaries.}
\end{itemize}

\newpage

\section{Preliminaries: Rotational and permutational symmetries for qubit systems}\label{App:A}

\subsection{Irreducible representations of SU(2) on $n$ qubits}

First, we recall a few useful facts about the representation theory of group $\SU(2)$. We are interested in the representation of $\SU(2)$ on the Hilbert space of $n$ qubits, where each unitary $U\in\SU(2)$ is represented by $U^{\otimes n}$. These unitary transformations describe, for instance, the global rotations of $n$ spin-half systems in 3D space. 

We consider an orthonormal basis in which this representation is manifestly decomposed into irreps. This basis can be defined in terms of the eigenvectors of two commuting observables, namely the operator $J_z\equiv \frac{1}{2}\sum_{j=1}^n \sigma^{(z)}_j$   and the total squared angular momentum operator $J^2=J^2_x+J^2_y+J_z^2$, also known as the Casimir operator. The eigenvalues of $J_z$ are $-n/2,\cdots,n/2$. The eigenvalues of $J^2$ are in the form $j(j+1)$, where $j$ takes values 
\begin{align}
  \text{Even } n:& \quad j_{\min} =0, 1,\cdots, \frac{n}{2}=j_{\max} \\
  \text{Odd } n:& \quad j_{\min}=\frac{1}{2}, \frac{3}{2},\cdots, \frac{n}{2}=j_{\max}\ . 
\end{align}
Each pair of eigenvalues $j(j+1)$ of $J^2$ and $m_z$ of $J_z$ has multiplicity $m(n, j)$. In particular, note that this multiplicity is independent of $m_z$.

Therefore, to decompose the representation of $\SU(2)$ into irreps we define the basis 
\begin{align}\label{tensor}
  |j, m_z, r\rangle \cong |j, m_z\rangle\otimes |j, r\rangle : \quad m_z=-j,\cdots, j-1, +j\ , \quad r=1,\cdots, m(n, j) ,
\end{align}
where $r=1,\cdots, m(n, j)$ is an index for the multiplicity of eigenvalues $j(j+1)$ of $J^2$ and $m_z$ of $J_z$. This basis is usually referred to as the Schur basis. \Cref{tensor} implies that the subspace of states with the total angular momentum $j$, denoted by $ \mathcal{H}^{(n)}_j$, has a tensor product decomposition as $\mathcal{H}^{(n)}_j\cong \mathbb{C}^{2j+1}\otimes \mathbb{C}^{m(n, j)}$, where $ \mathbb{C}^{2j+1}$ corresponds to the irrep of $\SU(2)$ with the total angular momentum $j$. To summarize, under the action of $U^{\otimes n}: U\in\SU(2)$, the Hilbert space of $n$ qubits decomposes as 
\begin{align}\label{Schur}
  (\mathbb{C}^2)^{\otimes n}\cong \bigoplus_{j=j_{\min}}^{j_{\max}} \mathcal{H}^{(n)}_j&=\bigoplus_{j=j_{\min}}^{j_{\max}} \mathbb{C}^{2j+1}\otimes \mathbb{C}^{m(n, j)}\equiv \bigoplus_{j=j_{\min}}^{j_{\max}} \mathbb{C}^{2j+1}\otimes \mathcal{M}_j\ ,
\end{align}
where $\mathcal{M}_j= \mathbb{C}^{m(n, j)}$ is called the multiplicity subsystem. Unitaries $U^{\otimes n}$ for $U\in\SU(2)$ act trivially on $\mathcal{M}_j$ and act irreducibly on $\mathbb{C}^{2j+1}$. In the rest of the paper, we often  suppress the superscript $n$, and write  $\mathcal{H}^{(n)}_j$ as $\mathcal{H}_j$. 

Consider an arbitrary state of $n$ qubits in the subspace with angular momentum $j$. Then, the reduced state of any $n-1$ qubits, can have components only in the subspaces with angular momenta $j\pm 1/2$. In other words,
\be
\mathcal{H}^{(n)}_j \subset \big(\mathcal{H}^{(n-1)}_{j-\frac{1}{2}}\oplus \mathcal{H}^{(n-1)}_{j+\frac{1}{2}}\big) \otimes \mathbb{C}^2\ ,
\ee
and $\mathcal{H}^{(n-1)}_{j\pm 1/2}$ are subspaces of $(\mathbb{C}^2)^{\otimes (n-1)}$ with angular momenta $j\pm 1/2$. Furthermore, the multiplicity of angular momentum $j$ in the system with $n$ qubits is equal to the sum of the multiplicities of angular momenta $j\pm 1/2$ in the system with $n-1$ qubits, i.e.,
\be\label{recursive}
m(n, j)=m(n-1, {j+1/2})+m({n-1}, {j-1/2})\ .
\ee
Therefore, the multiplicity subsystem can be decomposed as 
\begin{align}
  \mathbb{C}^{m(n, j)}\cong \mathbb{C}^{m(n-1, {j+\frac{1}{2}})}\oplus \mathbb{C}^{m(n-1, {j-\frac{1}{2}})} = \mathcal{M}_{j,+}\oplus \mathcal{M}_{j,-}\ ,
\end{align}
where
\be
\mathcal{M}_{j,\pm} \equiv \mathbb{C}^{m(n-1, {j\pm\frac{1}{2}})}\ .
\ee

The solution to the recursion relation \cref{recursive} is
\be\label{eq:hook_length}
m(n,j)={\binom{n}{\frac{n}{2}-j}}\times \frac{2j+1}{\frac{n}{2}+j+1}.
\ee
This multiplicity can also be calculated using the hook-length formula \cite{fulton2013representation} (see also Eq. (3.21) in \cite{QRF_BRS_07}). 

\subsection{Characterizing operators invariant under all rotations and permutations via Schur-Weyl duality}\label{app:duality}

Schur's lemmas imply that any rotationally-invariant operator  $A$ on $n$ qubits, i.e., an operator that commutes with $U^{\otimes n}$ for all $U \in \SU(2)$, has a decomposition as 
\be
A \cong \bigoplus_j (\identity_{2j+1} \otimes a_j)\ ,
\ee
 where $\identity_{2j+1}$ is the identity operator on $(2j+1)$-dimensional irrep of SU(2) and each $a_j$ is an operator on $m(n,j)$-dimensional subsystem $\mathcal{M}_j$, corresponding to the multiplicity of angular momentum $j$ of $n$ qubits. In particular, if $A$ is unitary, then $\{a_j\}$ are all unitaries.

  An example of rotationally-invariant unitaries are permutations on $n$ qubits, which correspond to  a representation of $\mathbb{S}_n$, the permutation group of $n$ objects, also known as the symmetric group. This means that in the Schur basis defined in the previous section, any  permutation $\sigma\in\mathbb{S}_n $ is represented as 
\be
\P(\sigma) \cong \bigoplus_j \big(\identity_{2j+1} \otimes \mathbf{p}_j(\sigma)\big)\ ,
\ee
where unitaries $\mathbf{p}_j$ define a unitary representation of $\mathbb{S}_n$ on  $m(n,j)$-dimensional subsystem $\mathcal{M}_j$. 

Indeed, according to Schur-Weyl duality there is a stronger relation between these representations of 
$\mathbb{S}_n$ and SU(2): The unitary representations $\{\mathbf{p}_j\}$ of $\mathbb{S}_n$ are all irreducible and inequivalent with each other. This, in particular, means any permutationally-invariant operator $B$ on $n$ qubits has a decomposition as 
\be
B \cong \bigoplus_j (b_j \otimes \mathbb{I}_{m(n,j)})\ ,
\ee
that is, they are block-diagonal with respect to subspaces with different angular momenta $j$, and they act as the identity operator on the irrep of $\mathbb{S}_n$, but they can be arbitrary on the irreps of SU(2).

Combining the above results about rotationally and permutationally-invariant operators, we conclude that operators that are invariant under both of these symmetries should be block-diagonal with respect to sectors $\{\mathcal{H}^{(n)}_j\}$ with different angular momenta. Furthermore, inside each subspace $\mathcal{H}^{(n)}_j\cong \mathbb{C}^{2j+1}\otimes \mathcal{M}_j$, they  should  be proportional to the identity operator  on both subsystems $ \mathbb{C}^{2j+1}$  and $\mathcal{M}_j$, which correspond to an irrep of SU(2) and $\mathbb{S}_n$, respectively.  It follows that  inside $\mathcal{H}^{(n)}_j$ they should be proportional to the identity operator.  In conclusion, any operator that is invariant under both of these symmetries can be written as a linear combination of projectors to subspaces $\{\mathcal{H}^{(n)}_j\}$, which are denoted by 
$\{\Pi_j\}$. In summary,  we find 
\begin{align}\label{AppdefC}
\mathcal{C}&\equiv \text{span}_\mathbb{C}\{\Pi_j: j=j_\text{min},\cdots , j_\text{max}\} =\Big\{A\in \End((\mathbb{C}^2)^{\otimes n}): [A, U^{\otimes n}]=[A, \P(\sigma)]=0\ ,  \forall\sigma\in\mathbb{S}_n , \forall U\in \text{SU}(2) \Big\}\ ,
\end{align}
where $\End((\mathbb{C}^2)^{\otimes n})$
 denotes the space of linear operators on $(\mathbb{C}^2)^{\otimes n}$.  Note that the dimension of $\mathcal{C}$ is $\lfloor n/2\rfloor+1$ and operators $\{\Pi_j\}$ form an orthogonal basis for it, that is  $\Tr(\Pi_j \Pi_{j'})=\delta_{j,j'} \Tr(\Pi_j)$.


\newpage

\section{Properties of  basis $\{C_l\}$ and polynomials $\{c_l(j)\}$ }\label{App}

In this section we  study properties of basis $\{C_l\}$ and their eigenvalues, namely  polynomials $\{c_l(j)\}$.  Recall the definition $C_0=\mathbb{I}$ and 
\be\label{art}
C_l\equiv \frac{1}{(l/2)!}  \sum_{i_1\neq i_2\neq\cdots \neq i_{l}} R_{i_1,i_2}\cdots\cdot R_{i_{l-1},i_{l}} \ \   \ \   \ \   \ \  :  l=2, 4 \cdots, 2\lfloor \frac{n}{2}\rfloor\ ,
\ee
where $i_1, i_2, \cdots, i_{l}$ are $l$ distinct integers between 1 to $n$, and  
\be
R_{rs}=\frac{1}{2}(\sigma^{(x)}_r \sigma^{(x)}_s+\sigma^{(y)}_r \sigma^{(y)}_s+\sigma^{(z)}_r \sigma^{(z)}_s)\ .
\ee
 Note that each operator $R_{i_1,i_2}\cdots\cdot R_{i_{l-1},i_{l}}$ appears exactly $2^{l/2}(l/2)!$ times in this above summation.   By definition, operator $C_l$ is   permutationally invariant. Furthermore, since each term $R_{rs}$ is rotationally invariant, Hermitian operator $C_l$ also enjoys rotational symmetry and therefore is in the  subspace $\mathcal{C}$, defined in Eq.(\ref{AppdefC}).  The presence of both rotational and permutational symmetries means that $C_l$ is block-diagonal with respect to sectors with different angular momenta, and can be written as a linear combination of projectors to these sectors, that is  
\be
C_l=\frac{1}{(l/2)!} \sum_{i_1\neq i_2\neq\cdots \neq i_{l}} R_{i_1,i_2}\cdots\cdot R_{i_{l-1},i_{l}}=\sum_{j} c_{l}(j)\ \Pi_j\  \ \   \ \   \ \   \ \  :  l=2, 4 \cdots, 2\lfloor \frac{n}{2}\rfloor\ ,
\ee
 where eigenvalues 
 \be
 c_{l}(j)=\frac{\Tr(C_l\Pi_j)}{\Tr(\Pi_j)}\ 
 \ee
  are real.  Indeed, operators $C_l: l=0,2, \cdots , 2\lfloor \frac{n}{2}\rfloor$ form an orthogonal basis for $\mathcal{C}$. In particular, 
 \be\label{OrthC}
  \Tr(C_l C_{l'})=\sum_{j} \Tr(\Pi_j)c_{l}(j) c_{l'}(j)  =\Tr(C^2_l)\  \delta_{l,l'}  \ .
\ee
This can be seen using the fact that  for operator $R_{rs}$ the partial trace over qubit $r$ (and, $s$) is zero, that is $\Tr_{r}(R_{rs})=0$. It follows that the Hilbert-Schmidt inner product of  $R_{i_1,i_2}\cdots\cdot R_{i_{l-1},i_{l}}$ and  $R_{k_1,k_2}\cdots\cdot R_{k_{l'-1},k_{l'}}$ can be non-zero only if they act on exactly the same set of qubits, that is $\{i_1,i_2, \cdots, i_{l}\}=\{k_1,k_2, \cdots, k_{l'}\}$. Then, assuming  $i_1\neq \cdots \neq i_l$, and $k_1\neq \cdots \neq k_{l'}$, this means  the inner product is non-zero only if $l=l'$, which in turn implies Eq.(\ref{OrthC})  (Indeed, one can show that if they both act on exactly the same set of qubits, then the inner product of  $R_{i_1,i_2}\cdots\cdots R_{i_{l-1},i_{l}}$ and  $R_{k_1,k_2}\cdots R_{k_{l'-1},k_{l'}}$ is always a positive number).

Eq.(\ref{OrthC}) implies 
\begin{align}
\Pi_j&=\sum_{l=0}^{2\lfloor{\frac{n}{2}}\rfloor} \frac{\Tr(\Pi_j C_l)}{\Tr(C_l^2)} C_l  \ \ \ \ \ \ : \ j=j_{\min}, j_{\min}+1. \cdots, j_{\max}\ .
\end{align}
\begin{remark}\label{rem0}
For a system with $n$ qubits, we have defined  operators  $C_l$  for even integers  $l=0, 2, \cdots, 2\lfloor n/2\rfloor$.  In the following, it is sometimes convenient to extend this definition to arbitrary even integers $l\ge 0$, by defining $C_l=0$ for $l>2\lfloor n/2\rfloor$. 
\end{remark}

Suppose we partition a system with $n$ qubits to two subsystems $A$ and $B$, with $n_A$ and $n_B$ qubits, where   $n_A+n_B=n$.  Let $C_l^{(AB)}$ be the $C_l$ operator on the total system,   
and $C_l^{(A)}$ and $C_l^{(B)}$ be the $C_l$ operators defined on the subsystems $A$ and $B$, respectively.  Then, using the fact that  for operator $R_{rs}$ the partial trace over qubit $r$ (and, $s$) is zero, and for single-qubit identity operator $I$, $\Tr(I)=2$,  one can  check that
\be\label{partial0}
\Tr_A(C^{(AB)}_l)= 2^{n_A} \ C^{(B)}_l\ \ \ \ , \  \text{and}\ \ \ \ \   \Tr_B(C^{(AB)}_l)= 2^{n_B}\  C^{(A)}_l\ .
\ee
We use this property later in Appendix \ref{App:l-body} to study the properties of $l$-body phases.

In the following, we first derive an explicit formula for the eigenvalues  $c_l(j)$ (See Eq.(\ref{eqcl})). Then, we show that this function is a polynomial of degree $l/2$ of $j(j+1)$.   Equivalently, we show that operator $C_l$ can be expressed as a polynomial of degree $l/2$ of $J^2$ (See Sec.\ref{Sec:pol}).

 \subsection{Determining eigenvalues $c_l(j)$ of operator $C_l$  (proof of Eq.\ref{clj1})}
 
 Next, we show that these eigenvalues should be indeed integer. Let $|\xi\rangle=(|0\rangle|1\rangle-|1\rangle|0\rangle)/\sqrt{2}$ be the singlet state of a pair of qubits.  Note that for  $ j=j_\text{min}, \cdots, j_\text{max}$ it holds that 
\be
\Pi_j \Big(|1\rangle^{\otimes (2j)}\otimes |\xi\rangle^{\otimes (n/2-j)}\Big)=|1\rangle^{\otimes (2j)}\otimes |\xi\rangle^{\otimes (n/2-j)}\  , 
\ee
that is, the state $|\eta_j\rangle=|1\rangle^{\otimes (2j)}\otimes |\xi\rangle^{\otimes (n/2-j)}$ lives in the subspace with angular momentum $j$, which implies
\be
c_{l}(j) =\langle \eta_j| C_l|\eta_j\rangle\ = \frac{1}{(l/2)!} \sum_{i_1\neq i_2\neq\cdots \neq i_{l}} \langle \eta_j|(R_{i_1,i_2}\cdots\cdot R_{i_{l-1},i_{l}}) |\eta_j\rangle\ .
\ee
Since each operator
$R_{i_1,i_2}\cdots\cdot R_{i_{l-1},i_{l}}$ appears exactly $2^{l/2}(l/2)!$ times in this summation, then to show that $c_{l}(j) $ is an integer, it suffices to show that for any choice of $l$ distinct integers  $i_1, i_2\cdots , i_{l}\in\{1,\cdots, n\} $, the expectation value
\begin{align}
2^{l/2}\times \langle \eta_j| (R_{i_1,i_2}\cdots R_{i_{l-1},i_{l}}) |\eta_j\rangle&= 2^{l/2}\times \Big(\langle1|^{\otimes (2j)}\otimes \langle \xi|^{\otimes (n/2-j)}\Big) \Big(R_{i_1,i_2}\cdots R_{i_{l-1},i_{l}}\Big) \Big(|1\rangle^{\otimes (2j)}\otimes |\xi\rangle^{\otimes (n/2-j)}\Big)\\ &=\Big(\langle1|^{\otimes (2j)}\otimes \langle \xi|^{\otimes (n/2-j)}\Big) \Big((2R_{i_1,i_2})\cdots (2R_{i_{l-1},i_{l}})\Big) \Big(|1\rangle^{\otimes (2j)}\otimes |\xi\rangle^{\otimes (n/2-j)}\Big)
\end{align}
is an integer. This follows immediately from the fact that
$2 R=\vec{\sigma}\cdot\vec{\sigma}$ satisfies
\begin{align}
\langle11| (2 R)|11\rangle&=1\ \ ,  \ \ \ \ \ 
\langle \xi| (2 R)|\xi\rangle=-3\ ,\\
\langle \xi|(I_A\otimes R_{BC})|\xi\rangle_{AB}&=0\ , \\
\langle \xi|(R_{AD}\otimes R_{BC})|\xi\rangle_{AB}&=-2 R_{CD}\ .
\end{align}
We conclude that coefficients $\{c_{l}(j) \}$ are all integer. 

Although the above method can also be used to calculate these integer coefficients,  here we use a different method. Using the fact that $C_l=\sum_j c_{l}(j) \ \Pi_j\ $, 
and $\Tr(\Pi_j)=m(n,j)\times (2j+1)$ 
we find
\be\label{clj}
c_{l}(j) =\frac{\Tr(\Pi_j C_l )}{\Tr(\Pi_j)}=\frac{\Tr(\Pi_j C_l) }{m(n,j)\times (2j+1)}\ .
\ee
Furthermore, 
\be\label{art2}
\Tr(\Pi_j C_l)= \frac{1}{(l/2)!}  \sum_{i_1\neq i_2\neq\cdots \neq i_{l}} \Tr(\Pi_j R_{i_1,i_2}\cdots\cdot R_{i_{l-1},i_{l}})= \frac{1}{(l/2)!} \times \frac{n!}{(n-l)!}\Tr(\Pi_j R_{1, 2}\cdots R_{{l-1},{l}})\ ,
\ee
where we have used the fact that  $\Pi_j$ is permutationally-invariant and $i_1,\cdots , i_{l}$ takes $n!/(n-l)!$ distinct values in this summation. Then, we use the fact that for a pair of qubits 
\be
\vec{\sigma}\cdot \vec{\sigma}=(I\otimes I)-4|\xi\rangle\langle\xi|\ ,
\ee
where $I$ denotes the single-qubit identity. This implies
\be
Q_l\equiv R_{1,2}\cdots R_{l-1,l}=  2^{-l/2}(\mathbb{I}-4|\xi\rangle\langle\xi|_{1,2})\cdots (\mathbb{I}-4|\xi\rangle\langle\xi|_{l-1,l})\ ,
\ee
where $\mathbb{I}=I^{\otimes n}$ is the identity operator on the Hilbert space of $n$ qubits, and $|\xi\rangle\langle\xi|_{i,j}$ denotes the projector to the singlet on qubits $i$ and $j$, tensor product with the identity operators on the rest of qubits. 

Then, using the binomial expansion, 
we find
\begin{align}\label{tar}
&\Tr(\Pi_j R_{1,2}\cdots R_{l-1,l})=2^{-l/2} \sum_{r=0}^{l/2} {{l/2}\choose{r}}  (-4)^{r}\Tr(\Pi_j \big[|\xi\rangle\langle\xi|^{\otimes r}\otimes I^{\otimes (n-2r)}\big])\ ,
\end{align}
where we again have used the fact that $\Pi_j$ is permutationally-invariant.  To calculate $
\Tr(\Pi_j \big[|\xi\rangle\langle\xi|^{\otimes r}\otimes I^{\otimes (n-2r)}\big]) $ we use the fact that 
for any state $|\phi\rangle$ of $n-2r$ qubits in the sector with angular momentum $j$, the $n$-qubit state $|\xi\rangle^{\otimes r}\otimes |\phi\rangle$ is in the sector with angular momentum $j$ of $n$ qubits. This implies
\be
\Tr(\Pi_j \big[|\xi\rangle\langle\xi|^{\otimes r}\otimes I^{\otimes (n-2r)}\big])=\Tr(\tilde{\Pi}_j)=m(n-2r,j)\times(2j+1)\ ,
\ee
where $\tilde{\Pi}_j$ is the projector to the subspace with angular momentum $j$ of $2r$ qubits. Putting this into Eq.(\ref{tar}) we conclude that
\be\label{jfd}
\Tr(\Pi_j Q_l)=\Tr(\Pi_j R_{1,2}\cdots R_{l-1,l})=2^{-l/2}
(2j+1)\sum_{r=0}^{l/2} {{l/2}\choose{r}}  (-4)^{r} \times m(n-2r,j)\ .
\ee
Combining this with Eq.(\ref{art2}) and Eq.(\ref{clj}) we conclude that the integer $c_{l}(j) $ is equal to 
\bes\label{eqcl}
\begin{align}
c_{l}(j) &=\frac{\Tr(\Pi_j C_l) }{m(n,j)\times (2j+1)} \\ &=\frac{1}{m(n,j)\times (2j+1)} \times \frac{1}{(l/2)!} \times \frac{n!}{(n-l)!}\Tr(\Pi_j R_{1, 2}\cdots R_{{l-1},{l}})\\ &=\frac{n!}{(n-l)! \times 2^{l/2}(l/2)!}  \sum_{r=0}^{l/2} (-4)^{r}  {{l/2}\choose{r}}   \frac{m(n-2r,j)}{m(n,j)}\ .
\end{align}
\ees
Expanding the binomials we find the formula
\begin{align}\label{formula:app}
c_{l}(j) &=\frac{l!}{2^{l/2}(l/2)!}{{n}\choose{l}}  \sum_{r=0}^{l/2} (-4)^{r}  {{l/2}\choose{r}}   \frac{m(n-2r,j)}{m(n,j)}\\ &=  \frac{1}{2^{l/2} (n - l)!} \sum_{r = 0}^{l / 2} \frac{(-4)^r r! (n - 2 r)! }{(l / 2 - r)!} \binom{\frac{n}{2}+j}{r} \binom{\frac{n}{2}-j}{r} \frac{\frac{n}{2} + j + 1}{\frac{n}{2} + j + 1 -r}\ .
 \end{align}
Eq.(\ref{jfd}) determines the inner product of $Q_s$ with $\Pi_j$. For completeness, here we also determine the inner product of $Q_s$ with $C_l$:
\begin{align}
 \Tr(Q_s C_l)=\Tr(R_{12} \cdots R_{s-1,s} C_{l} )&=  \frac{(n-s)!}{n!} \sum_{i_1\neq \cdots \neq i_s}\Tr(R_{i_1i_2}\cdots R_{i_{s-1}i_{s}} C_{l} ) 
 \\ &= \frac{(n-s)!(s/2)! }{n!} \Tr( C_s C_{l} )\\ &= \delta_{l,s} \frac{(n-s)!(s/2)!}{n!}  \Tr(C^2_{l}) \ .
\end{align}
We conclude that $Q_s$ has a positive overlap with $C_s$, whereas its overlap with $\{C_l: l\neq s \}$ is zero.

\subsection{$c_l(j)$ is a polynomial of degree $l/2$ of $j(j+1)$}\label{Sec:pol}

In the following, we present  a simple recursive relation for operator $\{C_l\}$. This relation, in particular, implies that $C_l$ can be expressed as a polynomial of degree $l/2$ of $J^2$.   This, in turn, implies that $c_l(j)$ is a polynomial of degree $l/2$ of $j(j+1)$. 

First, note that all operators $\{C_l\}$ commute with each other and are linear combinations of $\{\Pi_j\}$ operators. This means that the product of any pair of operators in this set  can be written as a linear combination of $\{\Pi_j\}$ operators, and hence a linear combination of  $\{C_l\}$.  In other words, operators $\{C_l\}$ form a commutative algebra.  This, in particular, implies that for any $l=0,\cdots, 2\lfloor n/2\rfloor$ it holds that 
\be
C_2 C_l= C_l C_2 =\sum_{l=0}^{2\lfloor{\frac{n}{2}}\rfloor} \alpha_{l}\  C_l \ ,
\ee
where $\alpha_l$ are real coefficients. Since $C_2$ and $C_l$  are sums of 2-local and $l$-local terms respectively, $C_2 C_l$ can be written as a sum of $(l+2)$-local operators. Furthermore, since $C_l$ is orthogonal to all $k$-local operators with $k<l$,  $C_2 C_l$  is orthogonal to all  $k$-local operators with $k<l-2$. We conclude that for $l$ in the interval $2, \cdots, 2\lfloor n/2\rfloor-2$, it holds that 
\be\label{alphas}
C_2 C_l= C_l C_2 = \alpha_{l-2}\ C_{l-2}+ \alpha_l\ C_l  + \alpha_{l+2} \ C_{l+2}\ 
\ee
for real numbers $\alpha_{l-2}, \alpha_{l}$, and  $\alpha_{l+2}$, which can be determined via  combinatorial arguments (See Sec.\ref{sec:rec}). It can be easily seen that for all $l$ in the interval $2, \cdots, 2\lfloor n/2\rfloor-2$, coefficient   $\alpha_{l+2}$ is non-zero (and positive). 
Therefore,  the above equation can be rewritten  as 
\be
C_{l+2}= \frac{1}{\alpha_{l+2} } C_2 C_{l} - \frac{\alpha_{l-2}}{\alpha_{l+2} }  C_{l-2}- \frac{\alpha_{l}}{\alpha_{l+2} }  C_{l}\ .
\ee
  
Next, note that 
\be
J^2=J_x^2+J_y^2+J_z^2=\frac{1}{4}(\sum_{i=1}^n \sigma^{(x)}_i)^2+\frac{1}{4}(\sum_{i=1}^n \sigma^{(y)}_i)^2+\frac{1}{4}(\sum_{i=1}^n \sigma^{(z)}_i)^2= \frac{3n}{4}+\sum_{i<j} R_{ij}= \frac{3n}{4}+\frac{1}{2} C_2\ ,
\ee
which implies
\be
C_2=2 J^2-\frac{3n}{2}\ .
\ee
Combining the above equations, we arrive at the recursive relation
\be
C_{l+2}=\frac{2}{\alpha_{l+2} }\  J^2 C_{l} - [\frac{\alpha_{l}+3n/2}{\alpha_{l+2} }]  \ C_{l} - \frac{\alpha_{l-2}}{\alpha_{l+2} }\  C_{l-2}\ ,
\ee
with the initial condition $C_0=\mathbb{I}$ and $C_2=2 J^2-\frac{3n}{2}\mathbb{I}$. It follows that $C_l$
 is a polynomial of degree $l/2$ in $J^2$.  In terms of the corresponding eigenvalues we obtain the recursive equation
\be
c_{l+2}(j)=\frac{2 j(j+1)}{\alpha_{l+2} } \ c_{l}(j)- [\frac{\alpha_{l}+\frac{3n}{2}}{\alpha_{l+2} }] \   c_{l}(j) - \frac{\alpha_{l-2}}{\alpha_{l+2} } \ c_{l-2}(j)\ .
\ee
 with the initial condition $c_0(j)=1$ and $c_2(j)=2 j(j+1)-\frac{3n}{2}$. Then,  again $c_l(j)$ is a polynomial of degree $l/2$ in $j(j+1)$.

  The following table lists the polynomials  $2^{l/2}(l/2)!\times c_l(j)$ for $l=0,\cdots, 10$.  \\

\begin{table*}[htb]
  \centering
  \renewcommand{\arraystretch}{1.2}
  \setlength{\tabcolsep}{4pt}
  \begin{tabular}{|c|l|}
    \TopRule
   {$l$} & \multicolumn{1}{c|}{$2^{l/2}(l/2)!\times c_l(j)$ expressed in terms of  $a=j(j+1)$} \\ \TopRule
    $0$ & $1$ \\ \MidRule
    $2$ & $4a - 3n$ \\ \MidRule
    $4$ & $16 a^2 - 8 (5n - 6) a + 15n(n-2)$ \\ \MidRule
    $6$ & $64 a^3 - 16 (21 n - 52) a^2 + 12 (35 n^2 - 154 n + 120)a - 105n (n-2) (n-4)$ \\ \MidRule
    $8$ & $256 a^4 - 256 (9 n - 34) a^3 + 288 (21 n^2 - 146 n + 216) a^2- 144 (35 n^3 - 336 n^2 + 892 n -560) a + 945 n (n-2) (n-4) (n-6)$ \\ \MidRule
    \multirow{2}{*}{$10$} & $1024 a^5 - 1280 (11 n - 56) a^4 + 128 (495 n^2 - 4730 n + 10224) a^3- 288 (385 n^3 - 5170 n^2 + 20680 n -23488) a^2$ \\
    & \multicolumn{1}{r|}{$+ 180 (385 n^4 - 6468 n^3 + 35948 n^2 - 73744 n + 40320)a-10395 n (n-2) (n-4) (n-6) (n-8)$} \\ \BotRule
  \end{tabular}
  \caption{Polynomial $2^{l/2}(l/2)!\times c_l(j)$ for $n$ qubits expressed as a function of $a=j(j+1)$}
  \label{tab:cl}
\end{table*}

\subsection{A recursive relation for polynomials $c_l(j)$}\label{sec:rec}
In the above argument we did not need to explicitly determine operator $C_2 C_l$, or, equivalently,  coefficients $\alpha_{l-2}, \alpha_l , \alpha_{l+2}$ in Eq.(\ref{alphas}). Nevertheless, for completeness sake we calculate them in the following. This calculation yields  a recursive formula for polynomials $c_l(j)$ (It should be noted that this calculation is not used in the rest of the paper).

We calculate
\begin{equation}\label{eq:alphas}
    C_2 C_l = \frac{1}{(l / 2)!} \sum_{i \neq j} \sum_{k_1 \neq \dots \neq k_l} R_{i, j} R_{k_1, k_2} \dots R_{k_{l - 1}, k_l}
\end{equation}
using the relations, with $i \neq j \neq k$,
\begin{align}
    R_{i, j}^2 & = \frac{3}{4} - R_{i, j}\ , \label{eq:Rsq} \\
    R_{i, j} R_{j, k} + R_{i, j} R_{i, k} & = \frac{1}{2} (R_{i, k} + R_{j, k})\ , \label{eq:Rtwo} \\
    R_{i, j} R_{i, k} R_{j, l} + R_{i, j} R_{j, k} R_{i, l} & = \frac{1}{2} R_{l, k} + \frac{1}{2} (R_{i, k} R_{j, l} + R_{i, l} R_{j, k} - 2 R_{i, j} R_{l, k})\ . \label{eq:Rthr}
\end{align}

The terms $R_{i, j} R_{k_1, k_2} \dots R_{k_{l - 1}, k_l}$ in the right-hand side of Eq.(\ref{eq:alphas}) can be grouped according to the relationships between $i \neq j$ and $\{k_m\}$: (1) $\{i, j\}$ and $\{k_m\}$ are disjoint; (2) $\{i, j\}$ and $\{k_m\}$ share a single element; (3) $i, j \in \{k_m\}$ but $R_{i, j}$ does not appear in $R_{k_1, k_2} \dots R_{k_{l - 1}, k_l}$; (4) $i, j \in \{k_m\}$ and $R_{i, j}$ does appear in $R_{k_1, k_2} \dots R_{k_{l - 1}, k_l}$. We approach each of these in kind.

The first terms, with $i \neq j \neq k_1 \neq \dots \neq k_l$, are grouped into the sum
\begin{equation}
    A_1 \equiv \frac{1}{(l / 2)!} \sum_{i \neq j \neq k_1 \neq \dots \neq k_l} R_{i, j} R_{k_1, k_2} \dots R_{k_{l - 1}, k_l} = \frac{(l / 2 + 1)!}{(l / 2)!} C_{l + 2} = (l / 2 + 1) C_{l + 2}.
\end{equation}

We write the second type of terms as a sum over $i = k_1 \neq j \neq k_2 \neq \dots \neq k_l$ and $i \neq j = k_1 \neq \dots \neq k_l$. These show up with multiplicity $2 (l / 2)$ since we are arranging for $k_1$ to be equal to either $i$ or $j$, hence
\begin{equation}
\begin{split}
    A_2 & \equiv \frac{2 (l / 2)}{(l / 2)!} \biggl[\sum_{i \neq j = k_1 \neq \dots \neq k_l} + \sum_{i = k_1 \neq j \neq k_2 \neq \dots \neq k_l}\biggr] R_{i, j} R_{k_1, k_2} \dots R_{k_{l - 1}, k_l} \\
    & = \frac{l}{(l / 2)!} \sum_{i \neq j \neq k_2 \neq \dots \neq k_l} (R_{i, j} R_{j, k_2} + R_{i, j} R_{i, k_2}) R_{k_3, k_4} \dots R_{k_{l - 1}, k_l} \\
    & = \frac{l}{(l / 2)!} \sum_{i \neq j \neq k_2 \neq \dots \neq k_l} \frac{R_{i, k_2} + R_{j, k_2}}{2} R_{k_3, k_4} \dots R_{k_{l - 1}, k_l} \\
    & = \frac{(n - l) l}{(l / 2)!} \sum_{i \neq k_2 \neq \dots \neq k_l} R_{i, k_2} R_{k_3, k_4} \dots R_{k_{l - 1}, k_l} \\
    & = (n - l) l C_l.
\end{split}
\end{equation}
In the third equality we use \cref{eq:Rtwo} and in the fourth we swap indices $i \leftrightarrow j$ for the second term and we sum over $j$.

The third terms of \cref{eq:alphas} we write as a sum over $i = k_1 \neq j = k_3 \neq \dots \neq k_l$ and $i = k_3 \neq j = k_1 \neq \dots \neq k_l$. Each has multiplicity $2^2 \binom{l / 2}{2}$, so
\begin{equation}
\begin{split}
    A_3 & \equiv \frac{4}{(l / 2)!} \binom{l / 2}{2} \biggl[\sum_{i = k_1 \neq j = k_3 \neq \dots \neq k_l} + \sum_{i = k_3 \neq j = k_1 \neq \dots \neq k_l}\biggr] R_{i, j} R_{k_1, k_2} \dots R_{k_{l - 1}, k_l} \\
    & = \frac{4}{(l / 2)!} \binom{l / 2}{2} \sum_{i \neq j \neq k_2 \neq k_4 \neq \dots \neq k_l} (R_{i, j} R_{i, k_2} R_{j, k_4} + R_{i, j} R_{j, k_2} R_{i, k_4}) R_{k_5, k_6} \dots R_{k_{l - 1}, k_l} \\
    & = \frac{4}{(l / 2)!} \binom{l / 2}{2} \sum_{i \neq j \neq k_2 \neq k_4 \neq \dots \neq k_l} \frac{1}{2} R_{k_2, k_4} R_{k_5, k_6} \dots R_{k_{l - 1}, k_l} \\
    & = \frac{4}{(l / 2)!} \binom{l / 2}{2} \binom{n - l + 2}{2} \sum_{k_2 \neq k_4 \neq \dots \neq k_l} R_{k_2, k_4} R_{k_5, k_6} \dots R_{k_{l - 1}, k_l} \\
    & = \frac{4}{l / 2} \binom{l / 2}{2} \binom{n - l + 2}{2} C_{l - 2} \\
    & = \binom{n - l + 2}{2} (l - 2) C_{l - 2}\ ,
\end{split}
\end{equation}
where in the third equality we used \cref{eq:Rthr} and the fact that the sums over the terms besides $R_{k_2, k_4} /2$ cancel.

Finally, with $i = k_1 \neq j = k_2 \neq \dots \neq k_l$, which has multiplicity $2 (l / 2)$,
\begin{equation}
\begin{split}
    A_4 & \equiv \frac{2 (l / 2)}{(l / 2)!} \sum_{i = k_1 \neq j = k_2 \neq \dots \neq k_l} R_{i, j} R_{k_1, k_2} \dots R_{k_{l - 1}, k_l} \\
    & = \frac{2 (l / 2)}{(l / 2)!} \sum_{i \neq j \neq k_3 \neq \dots \neq k_l} (3/4 - R_{i, j}) R_{k_3, k_4} \dots R_{k_{l - 1}, k_l} \\
    & = \frac{2 (l / 2)}{(l / 2)!} \biggl[\frac{3}{2} \binom{n - l + 2}{2} \sum_{k_3 \neq \dots \neq k_l} \bigl(R_{k_3, k_4} \dots R_{k_{l - 1}, k_l}\bigr) - (l / 2)! C_l\biggr] \\
    & = 3 \binom{n - l + 2}{2} C_{l - 2} - l C_l\ ,
\end{split}
\end{equation}
where in the second equality we used \cref{eq:Rsq}.

In total, we have
\begin{equation}
\begin{split}
    C_2 C_l & \equiv A_1 + A_2 + A_3 + A_4 \\
    & = (l / 2 + 1) C_{l + 2} + (n - l - 1) l C_l + \binom{n - l + 2}{2} (l + 1) C_{l - 2}\ ,
\end{split}
\end{equation}
which determines the coefficients in Eq.(\ref{alphas}), as 
\be
\alpha_{l-2}=(l + 1)\times \binom{n - l + 2}{2} \ , \ \ \  \ \ \alpha_l=(n - l - 1) l\ ,\   \ \ \ \ \alpha_{l+2}=l/2 + 1\ .
\ee 
With these coefficients, we can determine the recursion relation that defines the polynomials $c_l(j)$:
\begin{equation}
\begin{split}
    c_{l + 2}(j) = \frac{2 j (j + 1) - l (n - l - 1) - 3 n / 2}{l / 2 + 1}\times  c_l(j) - \binom{n - l + 2}{2} \frac{l + 1}{l / 2 + 1} \times  c_{l - 2}(j)\ .
\end{split}
\end{equation}
One can check that the polynomial $c_l(j)$ found in Eq.(\ref{formula:app}) indeed satisfies this recursive equation.

\newpage

\section{Full characterization of realizable rotationally-invariant Hamiltonians (Proof of theorem \ref{newThm})}\label{App:proofThem}

In this section we prove Theorem \ref{newThm}, which provides a full characterization of the family of Hamiltonians that can be realized using $k$-local rotationally-invariant unitaries:  for a system with $n$ qubits, the family of unitary evolutions $\exp({-i H t}): t\in\mathbb{R}$ generated by rotationally-invariant Hamiltonian $H$   
 can be implemented using $k$-local rotationally-invariant unitaries with $k\ge 2$, 
if, and only if
\begin{align}\label{con:app}
\Tr(H C_l)&=\sum_{j=j_\text{min}}^{n/2} c_l(j)\Tr(H\Pi_j)=0\ \  : \ \ \ l=2\lfloor{\frac{k}{2}}\rfloor+2,\cdots,  2\lfloor{\frac{n}{2}}\rfloor\ .
\end{align}

\noindent\textbf{Example:} As the simplest non-trivial example, consider  a system with $n=4$ qubits. Then, according to this theorem, the family of unitaries generated by rotationally-invariant   Hamiltonian $H$ can be implemented using the exchange interaction, up to a possible global phase, if and only if Eq.(\ref{con:app}) holds for $l=4$. Using Eq.(\ref{formula:app}), or, equivalently, using the polynomials in table \ref{tab:cl}, 
we can calculate $c_4(j)$  for a system with $n=4$ qubits. Then, we  obtain the necessary and sufficient condition 
\be
15 \Tr(H\Pi_{j=0})-5 \Tr(H\Pi_{j=1})+3 \Tr(H\Pi_{j=2})=0\ .
\ee

Let $\mathfrak{v}_{k}$ be the Lie algebra corresponding to the Lie group $\mathcal{V}_k$ generated by $k$-local rotationally-invariant unitaries. Equivalently, $\mathfrak{v}_{k}$ is the Lie algebra generated by $k$-local rotationally-invariant skew-Hermitian operators, i.e., 
\be\label{def:alg}
\mathfrak{v}_{k}\equiv \frak{alg}_{\mathbb{R}}\Big\{A \in \End((\mathbb{C}^2)^{\otimes n}):   A+A^\dag=0\ , A \text{ is $k$-local, }\ , [A, U^{\otimes n}]=0 : \forall U\in\text{SU}(2)  \Big\}\ ,
\ee
where $\frak{alg}_{\mathbb{R}}$ denotes the real Lie algebra generated by the operators, and 
$\End((\mathbb{C}^2)^{\otimes n})$
 denotes the space of linear operators on $(\mathbb{C}^2)^{\otimes n}$. 
 Then, in terms of Lie algebras,  theorem \ref{newThm} can be restated as 
\be
\mathfrak{v}_{k}=\Big\{A\in \mathfrak{v}_{n}: \Tr(A C_l)=0\ , l=2\lfloor{\frac{k}{2}}\rfloor+2,\cdots,  2\lfloor{\frac{n}{2}}\rfloor  \Big\}\ .
\ee

To establish this result, first 
we focus on the necessity of the constraints in Eq.(\ref{con:app}) on realizable Hamiltonians and then prove  the sufficiency of these constraints. Suppose Hamiltonian $H$ can be implemented using $k$-local rotationally invariant Hamiltonians, or more formally, suppose $i H \in\mathfrak{v}_k$ defined in Eq.(\ref{def:alg}). Then, clearly $H$ should  be rotationally invariant, i.e., $i H \in\mathfrak{v}_n$. To see that $H$ should also satisfy the constraints  in Eq.(\ref{con:app}), note that if $i H\in\mathfrak{v}_k$ then $i H$ can be written as a linear combination of the generators of $\mathfrak{v}_k$, namely, $k$-local skew-Hermitian operators $\{A_r\}$, and their nested commutators, as  
\be
i H=\sum_r \beta_r  A_r +\sum_{r_1, r_2} \beta_{r_1,r_2} [A_{r_1},  A_{r_2} ]+ \sum_{r_1, r_2, r_3} \beta_{r_1,r_2,r_3} [ [ A_{r_1},  A_{r_2} ], A_{r_3}]+\cdots \ ,
\ee
where $\beta_r, \beta_{r_1, r_2}, \cdots$ are  real coefficients and $\{A_r\}$ are $k$-local rotationally-invariant skew-Hermitian operators.  Consider the vector defined by $\Tr(H \Pi_j): j=j_\text{min},\cdots , j_\text{max}$ , called the charge vector of $H$ in \cite{marvian2022restrictions}. To determine this vector we note that 
\be
\Tr\big([A_{r_1},A_{r_2}] \Pi_j\big)= \Tr\big(A_{r_1}[A_{r_2}, \Pi_j]\big)=0\ ,
\ee
where the first equality follows from the cyclic property of trace, and the second equality is a consequence of the rotational symmetry of generators, which in particular means they commute with the projectors to sectors with different angular momenta.  It follows that
\be
 \Tr(H \Pi_j)= -i \sum_r \beta_r \Tr(A_r \Pi_j)\ \ :\ \ \ \ j=j_\text{min},\cdots , j_\text{max}\ . 
\ee
Recall that $\{C_l\}$ operators form another basis for the space of  
rotationally and permutationally invariant operators, denoted by $\mathcal{C}$, and in particular, each $C_l$ is a linear combination of projectors $\{\Pi_j\}$.  Therefore, the above equation can be equivalently restated as 
\be
\Tr(H C_l)=-i\sum_r \beta_r \Tr(A_r C_l)\ \ :\ \ \ \ l=0,\cdots , 2\lfloor\frac{n}{2}\rfloor\ .
\ee
Next,  recall that $C_l$ is orthogonal to $k$-local  operators with $k<l$. Since each generator $A_r$ is $k$-local, this means  $\Tr(C_l A_r)=0$ for $l>k$, We conclude that if $iH\in \mathfrak{v}_{k}$, then the constraints in Eq.(\ref{con:app}) hold. This completes the proof of the necessity of this condition, or, more precisely proves 
\be\label{subset}
\mathfrak{v}_{k}\subseteq \Big\{A\in \mathfrak{v}_{n}: \Tr(A C_l)=0\ , l=2\lfloor{\frac{k}{2}}\rfloor+2,\cdots,  2\lfloor{\frac{n}{2}}\rfloor  \Big\}\ .
\ee
Before presenting the argument for the sufficiency of this condition, we discuss a corollary of this result, which will be used in the next section. Consider the center of the Lie algebra $\mathfrak{v}_k$,  that is members of  $\mathfrak{v}_{k}$ commuting with the rest of this algebra, denoted by
\be
\mathfrak{z}_k \equiv \{A\in \mathfrak{v}_k: [A, F]=0 , \forall F\in \mathfrak{v}_k \}\ .
\ee
Using the above result we can show that
\begin{proposition}\label{prop:center}
The center of the Lie algebra $\mathfrak{v}_k$ is  $\mathfrak{z}_k=\text{span}_\mathbb{R}\{i C_l: l=0,\cdots, 2\lfloor k/2 \rfloor\}$. 
\end{proposition}
\begin{proof}
For $k=0$ and $k=1$  the  result holds trivially. Hence, in the following we assume $k\ge 2$.  To prove the claim, first note that  any operator $i C_l$ with $l\le k$ is a linear combination of $k$-local rotationally-invariant skew-Hermitian operators and therefore is  inside $\mathfrak{v}_k$. Furthermore, 
because $C_l=\sum_j c_l(j)\Pi_j$  it commutes  with all rotationally-invariant operators, and hence commutes with all elements of $\mathfrak{v}_k$. Therefore, it is in its center, $\mathfrak{z}_k$, that is
\be\label{center101}
 \text{span}\{i C_l: l=0, 2,\cdots,  2\lfloor k/2\rfloor\}\subseteq \mathfrak{z}_k \ .
\ee
  To see that the center $\mathfrak{z}_k$ is spanned by operators $iC_l:  l\le k$, we note that the elements of the center are both permutationally and rotationally invariant. This follows from the fact that for $k\ge 2$, $\mathfrak{v}_k$ contains 
$i\textbf{P}_{ab}: 1\le a<b\le n$, where 
$\textbf{P}_{ab}$ is the single-qubit swap operator that exchanges the states of qubits $a$ and $b$.  Clearly, $\textbf{P}_{ab}$ is 2-local and rotationally-invariant, which means $i\textbf{P}_{ab}\in\mathfrak{v}_k$ for $k\ge 2$. This immediately implies that the center of $\mathfrak{v}_k$ denoted by $\mathfrak{z}_k$ only contains rotationally-invariant skew-Hermitian operators that are also  permutationally-invariant (recall that any permutation can be obtained by a sequence of single-qubit swaps. Therefore, invariance under swaps implies permutational invariance). As we have seen before, such operators can be written as linear span of $\{i\Pi_j\}$ operators, or equivalently, the linear span of $\{i C_l\}$ operators. That is $\mathfrak{z}_k \subset \text{span}\{i C_l: l=0, 2,\cdots,  2\lfloor n/2\rfloor\}$. Combining this with Eq.(\ref{subset}) and using the orthogonality relation $\Tr(C_l C_{l'})=\delta_{l, l'} \Tr(C^2_l)$ we conclude that 
$\mathfrak{z}_k \subseteq \text{span}\{i C_l: l=0, \cdots, 2\lfloor k/2\rfloor\}$. 
 This together with Eq.(\ref{center101})   proves that $\mathfrak{z}_k = \text{span}\{i C_l: l=0, \cdots, 2\lfloor k/2\rfloor\}$, as stated in the above proposition. 
\end{proof}
In the next section, we discuss more about the Lie algebra $\mathfrak{v}_k$ and determine its dimension.

Next, we show the sufficiency of conditions in Eq.(\ref{con:app}). That is, we show that for any rotationally-invariant Hermitian operator $H$ satisfying these conditions, the family of unitaries $\exp(-i H t): t\in\mathbb{R}$ can be implemented using $k$-local rotationally-invariant unitaries. This amounts to showing that
$i H \in\mathfrak{v}_k$. As we saw in the main paper conditions in Eq.(\ref{con:app}) imply that $H$ can be written as 
\be\label{App: dec-22}
H=H_0+\sum_{l=0}^{2\lfloor n/2\rfloor} \frac{\Tr(H C_l)}{\Tr(C^2_l)}\ C_l\ =H_0+\sum_{l=0}^{2\lfloor k/2\rfloor}  \frac{\Tr(H C_l)}{\Tr(C^2_l)}\ C_l\ ,
\ee
where $H_0$ satisfies the condition 
\be\label{app:cent}
\Tr(H_0 \Pi_j)=\Tr(H_0 C_l)=0: \ \ l=0, 2,\cdots, 2\lfloor n/2\rfloor ; j=j_\text{min},\cdots , j_\text{max}\ .
\ee
Since the term $i(\sum_{l=0}^{2\lfloor k/2\rfloor}  \frac{\Tr(H C_l)}{\Tr(C^2_l)}\ C_l)$ can be written as a sum of $k$-local rotationally-invariant permutationally-invariant operators, it lives inside the center of $\mathfrak{v}_k$. 

Next, we focus on $iH_0\in\mathfrak{v}_k$ for $k\ge 2$. Note that Eq.(\ref{app:cent}) means that $iH_0$ does not have any component in the center of $\mathfrak{v}_k$. 

\begin{theorem}\label{Thm1}
Suppose rotationally-invariant Hamiltonian $H_0$ satisfies  $\Tr(H_0\Pi_j)=0$ for all $j=j_\text{min},\cdots, j_\text{max}$. Then, $i H_0\in \mathfrak{v}_k$ for $k\ge 2$, and therefore   the family of unitaries $\exp(-i H_0 t): t\in\mathbb{R}$, generated by Hamiltonian $H_0$ can be realized using 2-local rotationally invariant unitaries.  Equivalently, any rotationally-invariant unitary  $Y$ whose decomposition  $Y\cong \bigoplus_j \mathbb{I}_{2j+1}\otimes y_j$, in the Schur basis,   satisfies the property $\text{det}(y_j)=1:  j=j_\text{min}, \cdots, j_\text{max}$ can be realized using 2-local rotationally-invariant unitaries, i.e., $Y\in\mathcal{V}_2$.  
\end{theorem}
The equivalence of these two statements can be easily seen by noting that an operator $W$ is unitary and satisfies $\text{det}(W)=1$ if, and only if, there exists a traceless Hermitian operator $F$, such that $W=\exp(i F)$.    

It is useful to consider  other (equivalent) ways of expressing the above result.  Let $\mathcal{SV}_{k}\subset \mathcal{V}_{k}$ be the subgroup of $\mathcal{V}_{k}$ formed from rotationally-invariant unitaries $Y\cong \bigoplus_j \mathbb{I}_{2j+1}\otimes y_j$ that satisfy the additional constraint $\det(y_j)=1$ for all angular momenta $j$.    Then, the above theorem can be summarized as 
\be
\mathcal{SV}_{2}=\mathcal{SV}_{n}\ ,
\ee
or, equivalently, in terms of the corresponding Lie algebras
\be
\mathfrak{sv}_{2}=\mathfrak{sv}_{n}\ ,
\ee
where
\be
\mathfrak{sv}_{k}=\big\{A\in \mathfrak{v}_{k}: \Tr(A \Pi_j)=0\ , j=j_\text{min},\cdots,  j_\text{max}  \big\}\ .
\ee
Finally, note that any rotationally-invariant unitary $V\in \mathcal{V}_{n}$ has a unique decomposition as $V=[\sum_j \e^{\i \theta_j} \Pi_j] Y$, where $Y$ is in the group $\mathcal{SV}_{n}$ and $\theta_j\in[0,2\pi)$. Therefore, yet another way to phrase the above result is to say any rotationally-invariant unitary has a decomposition as   
\be
V=[\sum_j \e^{\i \theta_j} \Pi_j] Y\ :\ \ \ \ \   Y\in \mathcal{SV}_{2};\ \ \  \theta_j\in [0,2\pi)\ . 
\ee

   \subsection*{Proof of theorem \ref{Thm1}}
  
In the following, we present an overview of an elementary proof of theorem \ref{Thm1}, and present further technical details in \cite{HLM_2021}.  As stated before, we note that this result can also be shown using the more advanced Lie-algebraic results of Marin in \cite{marin2007algebre}.

 Recall the decomposition of the Hilbert space of $n$ qubits to irreps of SU(2):
 \be\label{decomp:SU}
(\mathbb{C}^2)^{\otimes n} \cong \bigoplus_{j=j_\text{min}}^{j_\text{max}} \mathcal{H}_j =\bigoplus_{j=j_\text{min}}^{j_\text{max}} \mathbb{C}^{2j+1} \otimes  \mathbb{C}^{m(n,j)} \ .
\ee
Relative to this decomposition, a general SU(2)-invariant unitary $V$ can be written as
\be\label{dec:app:V}
 V\cong\bigoplus_j \mathbb{I}_{2j+1} \otimes v_j\ ,
\ee
 where $v_j$   is an arbitrary unitary on the multiplicity subsystem $\mathbb{C}^{m(n,j)}$.

  To prove the theorem,  we use induction over $n$, the number of qubits.  That is we assume that for $n-1$ qubits, any rotationally-invariant  unitary $V$ can be realized, up to relative phases between sectors with different angular momenta.    From the representation theory of $\SU(2)$ we know that combining states with angular momentum $j'$ of $n-1$ qubits with states of a single qubit we obtain states of $n$ qubits with angular momenta $j'\pm\frac{1}{2}$. It follows that for $j\neq 0, \frac{n}{2}$ the multiplicity subsystem $\mathbb{C}^{m(n, j)}$ in \cref{decomp:SU} decomposes as 
\be\label{block}
\mathbb{C}^{m(n, j)} \cong \mathbb{C}^{m(n-1, j-\frac{1}{2})}\oplus \mathbb{C}^{m(n-1, j+\frac{1}{2})}\ ,
\ee
where $\mathbb{C}^{m(n-1, j\pm\frac{1}{2})}$ is the multiplicity subsystem associated to angular momentum $j\pm 1/2$ for a system with $n-1$ qubits.

According to the induction hypothesis, for any pair of unitaries $v_{j_\pm}$ acting on $\complex^{m(n - 1, j \pm \frac{1}{2})}$, there is an $\SU(2)$-invariant unitary $V$ acting on $n$ qubits with the properties that (i) $V$ can be realized by 2-local $\SU(2)$-invariant unitaries on the selected $n-1$ qubits and acts trivially on the remaining qubit, and (ii) it acts as the unitaries $v_{j_\pm}$ (up to phases) on the corresponding multiplicity subsystems $\mathbb{C}^{m(n-1, j\pm \frac{1}{2})}$ of the $n - 1$ qubits. Therefore, on the multiplicity space of the angular momentum $j$ sector of the $n$ qubits, $V$ acts as $v_{j_-} \oplus v_{j_+}$, i.e., is block-diagonal with respect to the decomposition in \cref{block}. Note that being block-diagonal with respect to this particular decomposition is a property of any $\SU(2)$-invariant unitary acting on the same $n - 1$ qubits, due to the fact that any such unitary separately preserves the angular momenta of the $n - 1$ qubits and the remaining qubit. On the other hand, this angular momentum does not remain conserved under general 2-local rotationally-invariant Hamiltonians that act non-trivially on qubit $n$, such as $R_{in}$ for $i\le n-1$. That is, unitaries generated by such Hamiltonians will not be block-diagonal with respect to the decomposition in \cref{block}. As we show in \cite{HLM_2021}, combining such unitaries with unitaries that are block-diagonal with respect to the decomposition in \cref{block} we obtain any desired unitary $v_j$ on the multiplicity subsystem $\mathbb{C}^{m(n, j)}$. More precisely, for any $V\in \mathcal{V}_{n}$ and angular momentum $j$, there exists $Y\in\mathcal{V}_{2}$ such that $V\Pi_j=Y \Pi_j$. 

The next step of the proof is to show that for unitaries generated by 2-local unitaries, all the corresponding unitaries $\{v_j\}$ in Eq.(\ref{dec:app:V})  can be chosen to be independent of each other; i.e., except a constraint on their phases, they do not satisfy any additional constraints among themselves. To establish this, we prove the following result, which is of independent interest.
For any pair of distinct qubits $a, b\in\{1,\cdots, n\}$ let $\textbf{P}_{ab}$ be the unitary that swaps the states of $a$ and $b$ and leave the other qubits invariant.

\begin{lemma}\label{lemma-main}
  Let $\mathcal{Y}$ be a subgroup of the group of rotationally invariant unitaries $\mathcal{V}_{n}$ satisfying the following properties: (i) for all angular momenta $j\le n/2$ the projection of $\mathcal{V}_{n}$ and $\mathcal{Y}$ to the sector $j$ are equal in the following sense: for any $V\in \mathcal{V}_{n}$ there exists $Y\in \mathcal{Y}$ and a phase $\e^{\i\theta}$ such that $ V\Pi_j=\e^{\i\theta} Y\Pi_j $; and (ii) for any pair of qubits $a$ and $b$, $\mathcal{Y}$ contains the swap ${\P}_{ab}$, up to a possible global phase $\e^{\i\theta}$, i.e. $\e^{\i\theta}{\P}_{ab}\in \mathcal{Y}$. Then $\mathcal{SV}_{n}\subset \mathcal{Y}$, or equivalently, any rotationally-invariant unitary $V\in \mathcal{V}_{n}$ can be written as $V=(\sum_j \e^{\i\theta_j}\Pi_j) Y$ for $Y\in\mathcal{Y}$ and $\theta_j\in [0,2\pi)$.
\end{lemma}
Intuitively, the first assumption means that inside each angular momentum sector, the unitaries in $\mathcal{Y}$ are not restricted, i.e., can be any rotationally invariant unitary, whereas the second assumption guarantees a certain level of independence between unitaries in different angular momentum sectors. Combining this lemma with the above argument, one can prove \cref{Thm1} (See \cite{HLM_2021} for further details). \\

\newpage

\section{The dimension of the Lie group generated by \texorpdfstring{$k$}{k}-local \texorpdfstring{$\SU(2)$}{SU(2)}-invariant unitaries (proof of Eq. (\ref{dim})) }\label{App:dim}

As we saw in theorem \ref{newThm},  a general rotationally-invariant Hamiltonian $H$  can be implemented using $k$-local rotationally-invariant unitaries,  i.e.,  $i H\in\mathfrak{v}_k$, if and only if 
\be
\Tr(H C_l)=0 \ : \ \ \  \ l=2\lfloor k/2\rfloor+2,\cdots, 2\lfloor n/2\rfloor \ ,
\ee
which means $H$ can be decomposed as
\be
H=H_0+\sum_{l=0}^{2\lfloor k/2\rfloor} \frac{\Tr(H C_l)}{\Tr(C^2_l)}\ C_l\ \ ,
\ee
where $H_0$ can be any rotationally-invariant Hermitian operator satisfying the condition 
\be
\Tr(H_0 \Pi_j)=\Tr(H_0 C_l)=0: \ \ l=0, 2,\cdots, 2\lfloor n/2\rfloor\ \ ; \ \ j=j_\text{min},\cdots , j_\text{max}\ .
\ee 
In terms of the corresponding  Lie algebras, this result can be restated as following: 
 \be\label{eq:Liealg}
\mathfrak{v}_{k}=\mathfrak{sv}_{k} \oplus \mathfrak{z}_{k}\ ,
\ee
where
\be
\mathfrak{v}_{k}\equiv \frak{alg}_{\mathbb{R}}\Big\{A: A+A^\dag=0\ , A \text{ is $k$-local, }\ , [A, U^{\otimes n}]=0 : \forall U\in\text{SU}(2)  \Big\}\ ,
\ee
the sub-algebra $\mathfrak{z}_k$ is 
\be
\mathfrak{z}_k=\text{span}_\mathbb{R}\{i C_l: l=0,\cdots, 2\lfloor k/2 \rfloor\}\ ,
\ee
which according to proposition \ref{prop:center} is the center of $\mathfrak{v}_{k}$, and 
  \be
\mathfrak{sv}_{k}\equiv \big\{A\in \mathfrak{v}_{k}: \Tr(A \Pi_j)=0\ , j=j_\text{min},\cdots j_\text{max}  \big\}\ ,
\ee
which according to theorem \ref{Thm1} satisfies 
\be
\mathfrak{sv}_{k}=\mathfrak{sv}_{2}=\mathfrak{sv}_{n}\ ,
\ee
for $k\ge 2$. Note that $\mathfrak{sv}_{k}$ and $\mathfrak{z}_{k}$ are orthogonal with respect to the  Hilbert Schmidt inner product (or, equivalently, with respect to the Killing form).

 The latter equation together with the decomposition in Eq.(\ref{eq:Liealg})  imply
\be
\dim(\mathfrak{v}_{k})=\dim(\mathfrak{z}_{k})+\dim(\mathfrak{sv}_{k})= \dim(\mathfrak{z}_{k})+\dim(\mathfrak{sv}_{n})=\dim(\mathfrak{z}_{k})+\dim(\mathfrak{v}_{n})-\dim(\mathfrak{z}_{n}) \ .
\ee
Note that all the above Lie algebras are 
vector spaces over the field of real numbers $\mathbb{R}$, and $\dim(\cdot)$ denotes the dimension of the vector space.  
 
Since operators $i C_l: l=0,\cdots, 2\lfloor k/2 \rfloor$ are linearly independent, then  
\be\label{dim:center}
\dim(\mathfrak{z}_{k})=\lfloor \frac{k}{2}\rfloor+1\ .
\ee
Then, we arrive at the identity 
\be
\dim(\mathfrak{v}_{k})=\dim(\mathfrak{v}_{n})+\lfloor \frac{k}{2}\rfloor-\lfloor \frac{n}{2}\rfloor\ ,
\ee
or, equivalently, in terms of the corresponding Lie groups
\be\label{dif13}
\dim(\mathcal{V}_{k})=\dim(\mathcal{V}_{n})+\lfloor \frac{k}{2}\rfloor-\lfloor \frac{n}{2}\rfloor\ .
\ee
It is worth noting that Eq.(\ref{dim:center}) that determines the dimension of the center of $\mathfrak{v}_k$ can also be obtained from a general result of \cite{marvian2022restrictions}. According to this result, for a connected Lie group $G$, such as SU(2), the dimension of the center of the Lie group generated by $k$-local $G$-invariant skew-Hermitian operators is  equal to $|\text{Irreps}_G(k)|$, that is the number of inequivalent irreps of $G$ appearing in the  representation of $G$ on $k$ subsystems. In the case of SU(2) symmetry with spin-half systems, the inequivalent irreps correspond to different angular momentum values  $j=0, 1,\cdots , k/2$ for even $k$ and 
$j=1/2, 3/2,\cdots , k/2$ for odd $k$.  Therefore, the number of inequivalent irreps  is  $\lfloor \frac{k}{2}\rfloor+1$, which is in agreement with  Eq.(\ref{dim:center}).

In the following, we calculate the dimension of $\mathfrak{v}_{n}$,  the Lie algebra of rotationally-invariant skew-Hermitian operators defined on $n$ qubits,  and prove 
\begin{equation}
  \dim(\mathfrak{v}_{n}) = \frac{1}{n+1}\binom{2n}{n}\ ,
\end{equation} 
which together with \cref{dif13} implies  \cref{dim}.  First, note that
\begin{align}
  \dim_\mathbb{R}(\mathfrak{v}_{n})&=\dim_\mathbb{R}\Big(\Big\{A\in \End((\mathbb{C}^2)^{\otimes n}):\ A+A^\dag=0, [A, U^{\otimes n}]=0,\ \forall U\in\SU(2) \Big\}\Big)\\ &=\dim_\mathbb{C}\Big(\Big\{A\in \End((\mathbb{C}^2)^{\otimes n}):\ [A, U^{\otimes n}]=0,\ \forall U\in\SU(2) \Big\}\Big) \ , 
\end{align}
where $\dim_\mathbb{C}$ and $\dim_\mathbb{R}$ denotes dimension as vector space over the field of complex and real numbers, respectively. 
Note that the vector space in the second line is the space of all rotationally-invariant operators on $(\mathbb{C}^2)^{\otimes n}$, and the equality follows from the fact that a general rotationally-invariant operator can be decomposed as the sum of a skew-Hermitian and a Hermitian rotationally-invariant operators. It is worth noting that the space of all rotationally-invariant operators on $(\mathbb{C}^2)^{\otimes n}$ can be interpreted as a complex Lie algebra, which is indeed the complexification of the real Lie algebra $\mathfrak{v}_{n}$, i.e., $\mathfrak{v}_{n}^\mathbb{C}\equiv \mathfrak{v}_{n} + \i \mathfrak{v}_{n}$. Using this notation, the above identity can be rewritten as
\be
\dim_\mathbb{R}(\mathfrak{v}_{n}) = \dim_\complex( \mathfrak{v}_{n}^\complex)\ .
\ee
Next, we calculate the dimension of the space of rotationally-invariant operators on $n$ qubits, i.e., $\big\{A\in \End((\mathbb{C}^2)^{\otimes n}):\ [A, U^{\otimes n}]=0,\ \forall U\in\SU(2) \big\}$, as a complex vector space.  Using the fact that $\SU(2)$ is self-dual, we can easily show that as representations of $\SU(2)$
\begin{equation}
  \End((\mathbb{C}^2)^{\otimes n}) \cong (\mathbb{C}^2)^{\otimes n} \otimes (\mathbb{C}^2)^{\otimes n\ast} \cong (\mathbb{C}^2)^{\otimes 2n} \cong \bigoplus_{j=0}^n \mathbb{C}^{2j+1}\otimes \mathbb{C}^{m(2n, j)}\ ,
\end{equation}
where by $(\complex^2)^{\otimes n\ast}$ we mean a vector space carrying a representation of $\SU(2)$ equivalent to $(U^\ast)^{\otimes n}$, where $U^\ast$ is the complex conjugate of $U$ in the $\{|0\rangle,|1\rangle\}$ basis.
To understand this identity, it is useful to consider the natural isomorphism between $\End((\mathbb{C}^2)^{\otimes n}) $ and $(\mathbb{C}^2)^{\otimes 2 n}$ defined by 
\be
v=\sum_{r_1\cdots r_n, s_1\cdots s_n=0}^1 v_{r_1\cdots r_n, s_1\cdots s_n} |r_1 \cdots r_n\> \< s_1 \cdots s_n| \mapsto \mathrm{vec}(v) = \sum_{r_1\cdots r_n, s_1\cdots s_n=0}^1 v_{r_1\cdots r_n, s_1\cdots s_n} \ |r_1 \cdots r_n\>|s_1 \cdots s_n\>\ ,
\ee
where we have used the qubit orthonormal basis $\{|0\rangle,|1\rangle\}$. Under the action of $\SU(2)$, $v \mapsto U^{\otimes n} v U^{\otimes n \dagger}$, and so 
\be
\operatorname{vec}(v) \mapsto \big[U^{\otimes n} \otimes U^{\ast \otimes n}\big] \operatorname{vec}(v) = \big[U^{\otimes n} \otimes (Y U Y)^{\otimes n}\big] \operatorname{vec}(v)\ ,
\ee
where we have used the fact that $\SU(2)$ is self-dual, and in particular $U^\ast=YUY$, where $Y=i(|1\rangle\langle 0|-|0\rangle\langle 1|)$ is the Pauli-y unitary matrix.  This establishes an isomorphism between the linear space of $\SU(2)$-invariant operators on $(\mathbb{C}^2)^{\otimes n}$ and vectors in sector $j=0$ of $\mathbb{C}^{\otimes 2n}$, which, in particular, implies, 
\be
\dim_\mathbb{R}(\mathfrak{v}_{n})=\dim_\mathbb{C}\Big\{ A\in\End((\mathbb{C}^2)^{\otimes n}): [A, U^{\otimes n}]=0, \forall U\in\SU(2) \Big\}= \dim_\complex( \mathfrak{v}_{n}^\complex)=m(2n,0)=\frac{1}{n+1}\binom{2n}{n}\ ,
\ee
where $m(2n,0)$ is the multiplicity of angular momentum $j=0$ for a system with $2n$ qubits, and to get the last equality we have used \cref{eq:hook_length}.  Together with \cref{dif13}, this proves \cref{dim}.

\newpage

\section{Average energy for a fixed angular momentum (Proof of Eq.(\ref{avg}))}\label{App:avg}

For any angular momentum $j=j_\text{min},\cdots, j_\text{max}$ the average energy for a fixed angular momentum is defined as the expectation value of Hamiltonian $H$ for the maximally-mixed state over the sector with angular momentum $j$, i.e., state $\Pi_j/\Tr(\Pi_j)$, as  
\be
E_j\equiv \frac{\Tr(H\Pi_j)}{\Tr(\Pi_j)}\ .
\ee
Using the relation between $\{\Pi_j\}$ basis and $\{C_l\}$ basis, namely,
\be
\Pi_j=\sum_{l=0}^{2\lfloor{\frac{n}{2}}\rfloor} \Tr(\Pi_j C_l) \frac{C_l}{\Tr(C_l^2)}= \Tr(\Pi_j)  \sum_{l=0}^{2\lfloor{\frac{n}{2}}\rfloor} c_l(j)  \frac{C_l}{\Tr(C_l^2)}\ .
\ee
This can be rewritten as 
\be\label{energy3}
E_j\equiv \frac{\Tr(H\Pi_j)}{\Tr(\Pi_j)}=\sum_{l=0}^{2\lfloor{\frac{n}{2}}\rfloor}  c_l(j) \times \frac{\Tr(H C_l)}{\Tr(C_l^2)}=\sum_{l=0}^{2\lfloor{\frac{n}{2}}\rfloor}  c_l(j) \times q_l \ ,
\ee
where $q_l={\Tr(H C_l)}/{\Tr(C_l^2)}$ and the summation is over even integer $l$. 

Now recall that $C_l$ is orthogonal to all $k$-local operators with $k<l$.  Therefore, if $H$ can be written as a sum of $k$-local operators, then $\Tr(H C_l)=0$ for $l>k$. It follows that Eq.(\ref{energy3}) can be written as
\be\label{energy47}
E_j\equiv \frac{\Tr(H\Pi_j)}{\Tr(\Pi_j)}=\sum_{l=0}^{2\lfloor{\frac{k}{2}}\rfloor}   c_l(j)   \times \frac{\Tr(H C_l)}{\Tr(C_l^2)}=\sum_{l=0}^{2\lfloor{\frac{k}{2}}\rfloor}  c_l(j) \times q_l\ .
\ee
Note that to derive this equation we did not make any assumptions about symmetries of $H$. However, in the presence of symmetries a stronger result holds: Suppose the family of unitaries $\exp(-i H t): t\in\mathbb{R}$ can be implemented using $k$-local rotationally invariant unitaries, which means $iH$ is in $\mathfrak{v}_k$, the Lie algebra of generated by $k$-local rotationally-invariant skew-Hermitian operators. Then, again $H$ is orthogonal to $C_l$ for $l>k$ and Eq.(\ref{energy47}) holds. 

According to Eq.(\ref{const3}) in theorem \ref{newThm},  for any such Hamiltonian $H$, we also have 
\be\label{energy4}
\sum_{j=j_\text{min}}^{n/2} c_l(j) \Tr(\Pi_j)E_j= \sum_{j=j_\text{min}}^{n/2} c_l(j)\Tr(H\Pi_j)=\Tr(H C_l)= 0\ :\ \ \ \ l=2\lfloor{\frac{k}{2}}\rfloor+2,\cdots,  2\lfloor{\frac{n}{2}}\rfloor\ .
\ee
It turns out that the above equations  are indeed equivalent. That is Eq.(\ref{energy4}) implies Eq.(\ref{energy47}) holds for some coefficients $q_l: 0,\cdots,  2\lfloor{\frac{k}{2}}\rfloor$. And  any function in the form $E_j=\sum_{l=0}^{2\lfloor{\frac{k}{2}}\rfloor}  c_l(j) \times q_l$ satisfies the constraints in Eq.(\ref{energy4}).  This equivalence can be explained using the orthogonality relations
\be
\sum_{j=j_\text{min}}^{n/2} \Tr(\Pi_j) c_l(j) c_{l'}(j) = \Tr(C_l C_{l'})=\delta_{l,l'} \times \Tr(C_l^2)\ .
\ee
Multiplying both sides of Eq.(\ref{energy47}) in $c_{l'}(j) \Tr(\Pi_j)$ and summing over $j$ we obtain Eq.(\ref{energy4}).  Conversely, the above orthogonality relation implies that functions $c_l: l=0, \cdots ,   2\lfloor\frac{n}{2}\rfloor$ form a basis for the space of functions defined on $\{j_\text{min}, j_\text{min}+1,\cdots, n/2\}$, which has dimension $\lfloor{{n}/{2}}\rfloor+1$. Furthermore, Eq.(\ref{energy4}) means that function $E_j$ is orthogonal to all $c_l: l=2\lfloor\frac{k}{2}\rfloor+2, \cdots ,   2\lfloor\frac{n}{2}\rfloor$. It follows that this function can be a written as a linear combination of functions $c_l: l=0, \cdots, 2\lfloor\frac{k}{2}\rfloor$, as stated in Eq.(\ref{energy47}).

\newpage

\section{Time-reversal symmetry and restrictions on realizable Hamiltonians}\label{App:TimeReversal}

In this section we show that locality does not restrict realizable Hamiltonians with odd parity under time-reversal symmetry.  We also explain why only even integer $l$ appear   in the basis $\{C_l: l=0, 2, \cdots , 2\lfloor\frac{n}{2}\rfloor \}$, that defines the notion of $l$-body phases.

Consider the  basis for the space of $n$-qubit operators defined by the $n$-fold tensor products of Pauli operators and the single-qubit identity operator, i.e.,  is the basis defined by 
\be
\mathcal{B}_{\text{Pauli}}=\{\sigma^{(x)},\sigma^{(y)}, \sigma^{(z)}, I\}^{\otimes n}\ ,
\ee
 where $I$ is the single-qubit identity operator.  Note that this is an orthogonal basis with respect to the Hilbert-Schmidt inner product.   The Pauli weight of any such tensor product is the number of Pauli operators that appear in the  product.

Recall that $J_r=\frac{1}{2}\sum_{i=1}^n \sigma^{(r)}_i$. This implies that operators $J_x^2$, $J_y^2$, and $J_z^2$, can all be written as a linear combination of elements of this basis with Pauli weights 0 and 2. For instance, 
\be
J_z^2=\frac{1}{4}\sum_{i,j=1}^n \sigma_i^{(z)} \sigma_j^{(z)}\ ,
\ee
where $\sigma_i^{(z)}$ denotes Pauli operator $z$ on qubit $i$ tensor product with the identity operators on the rest of qubits.
It follows that 
\be
J^2=J_x^2+J_y^2+J_z^2=\sum_j j(j+1) \Pi_j\ 
\ee
can also be written as a linear combination of elements of $\mathcal{B}_{\text{Pauli}}$ with Pauli weights 0 and 2. Since  projectors $\Pi_j$ can be written as a polynomial of $J^2$, we conclude that projectors $\Pi_j$ can be written as a linear combination of basis elements in   
$\{\sigma^{(x)},\sigma^{(y)}, \sigma^{(z)}, I\}^{\otimes n}$ with even Pauli weights. 
Furthermore, since different elements of this basis are orthogonal we conclude that 
\begin{proposition}
Suppose operator $A$ can be written as a linear combination of basis elements $\{\sigma^{(x)},\sigma^{(y)}, \sigma^{(z)}, I\}^{\otimes n}$ with odd Pauli weights. Then, $\Tr(A \Pi_j)=0$ for all $j=j_\text{min},\cdots, j_\text{max}$. Using theorem \ref{newThm}, we conclude that any rotationally-invariant Hamiltonian $A$ with this property can be realized using the exchange interaction alone.  
\end{proposition}
Note that  in the basis $\{C_l: l=0, 2, \cdots , 2\lfloor\frac{n}{2}\rfloor \}$, each operator $C_l$ can be written as a linear combination of elements of 
$\mathcal{B}_{\text{Pauli}}=\{\sigma^{(x)},\sigma^{(y)}, \sigma^{(z)}, I\}^{\otimes n}$ with even $l$, and such operators span $\mathcal{C}=\text{span}\{\Pi_j\}$.  Recall from proposition \ref{prop:center} that operators $\{C_l: l=0, 2, \cdots , 2\lfloor\frac{k}{2}\rfloor \}$ span 
the center of  $\mathfrak{v}_k$, the Lie algebra generated by $k$-local rotationally-invariant skew-Hermitian operators.   Therefore, these propositions together imply that a skew-Hermitian operator $A$ satisfying the above condition does not have any component in the center of $\mathfrak{v}_k$. 
 
 In the paper we saw that for systems with odd number of qubits $n$, universality can be achieved with $(n-1)$-local symmetric unitaries, whereas  this is not possible for systems with even number of qubits. The above result explains this observation:
If $n$ is odd, then operators that can be written as a linear combination of $n$-fold tensor product of Pauli operators, i.e., $\{\sigma^{(x)},\sigma^{(y)}, \sigma^{(z)}\}^{\otimes n}$, are orthogonal to all $\{C_l\}$ operators, or equivalently, to all  $\{\Pi_j\}$ operators. Therefore, such Hamiltonians can all be realized by the exchange interaction alone. This means that we can achieve universality with $(n-1)$-local rotationally-invariant unitaries.

Another interesting implication of the above observation is that locality does not constrain Hamiltonians with odd parity under time-reversal symmetry.   Let $\Theta$ be the anti-unitary time-reversal operator satisfying $\Theta i \Theta^{-1}=-i$ and 
\be
\Theta J_r \Theta^{-1}=-J_r  \ .
\ee
 This operator transforms any state to a state with the opposite angular momentum.  
 Clearly,  the total squared angular momentum operator remains invariant under the time reversal symmetry, which means its eigen-projectors $\{\Pi_j\}$  also remain invariant, that is 
 \be
 \Theta \Pi_j \Theta^{-1}=\Pi_j\ .
 \ee
  It follows that all Hermitian operators in $\mathcal{C}=\text{span}\{\Pi_j\}$
have this symmetry.   
 
  Let $H$ be a Hamiltonian with odd parity under the time reversal symmetry, such that $\Theta H \Theta^{-1}=-H$.  For such Hamiltonians $\Tr(H\Pi_j)=0$ for all angular momentum $j$, which in turn implies  $\Tr(H C_l)=0$ for all 
 $l=0, 2,\cdots ,  2 \lfloor{{n}/{2}}\rfloor$.  This means that this  family of Hamiltonian can be realized using the exchange interaction. In terms of the Lie algebras this means that,  all Hamiltonians with odd parity under time reversal symmetry belong to the  Lie algebra generated by 
this 2-local rotationally-invariant Hamiltonians. This provides another way to understand why $(n-1)$-local rotationally-invariant unitaries are universal for all symmetric unitaries when $n$ is odd.

\newpage

\section{l-body phases and their properties}\label{App:l-body}

In the section we discuss more about 
$l$-body phases and prove their properties listed in the paper.  These arguments extend  similar results about U(1) symmetry  presented in \cite{marvian2022restrictions}.

We start by adapting a general result of \cite{marvian2022restrictions} to the case of SU(2) symmetry.  Let $V$ be a rotationally-invariant unitary transformation on $n$ qubits, generated under rotationally-invariant Hamiltonian $H(t)$, from time $t=0$ to $T$. This unitary is given by the time-ordered integral  
\be\label{Sch}
V=\mathcal{T}\{\exp(-i\int_0^T H(t) dt)\}=\lim_{M\rightarrow \infty} \prod_{r=1}^M\exp(-i\frac{T}{M} H(\frac{r T}{M})) \ .
\ee
Let  $V_j$ be the component of $V$ in $\mathcal{H}_j$, the subspace with angular momentum $j$, such that $V\cong \bigoplus_j V_j$.  Consider an observable $C=\sum_j c(j) \Pi_j$ with integer  eigenvalues $c(j)$.  Then, 
\be\label{det}
\Phi\equiv\sum_j c(j)\ \theta_j=\sum_j c(j)\ \text{arg}(\text{det}(V_j))=   -\int_0^T dt \ \Tr(H(t) C)  \ \ \ \ \ \ \text{: mod} \ 2\pi\ ,
\ee
where we have defined  
\be\label{ary1}
\theta_j\equiv \text{arg}(\text{det}(V_j))=   -\int_0^T dt \ \Tr(H(t) \Pi_j)  \ \ \ \ \ \ \text{: mod} \ 2\pi\ .
\ee
Now  choosing operator $C$ to be
\be\label{hghgh1}
C_l\equiv \frac{1}{(l/2)!} \hspace{1mm}\sum_{i_1\neq \cdots \neq i_{l}} \hspace{-3mm} R_{i_1,i_2}\cdots\cdot R_{i_{l-1},i_{l}}  =\sum_{j=j_\text{min}}^{n/2} c_{l}(j) \ \Pi_j\ ,
\ee 
  we obtain
\be\label{def-App}
\Phi_l\equiv    \sum_{j=j_\text{min}}^{j_\text{max}} c_l(j)\ \theta_j= \sum_{j=j_\text{min}}^{j_\text{max}} c_l(j)\ \text{arg}(\text{det}(V_j))= -\int_0^T dt \ \Tr(H(t) C_l)\ \ \ \ \ \ \text{: mod} \ 2\pi\ .
\ee

The above equation defines $\Phi_l$, up to an integer multiple of $2\pi$. In the following, we often assume $\Phi_l\in(-\pi, \pi]$.  

The last equality in Eq.(\ref{def-App}) holds for any Hamiltonian $H(t)$ that realizes unitary $V$  under the Schr\"{o}dinger equation, i.e., satisfies Eq.(\ref{Sch}). In this sense the $l$-body phases of rotationally-invariant unitaries are path-independent.   
Note that the validity of this equality relies on the fact that eigenvalues  $\{c_l(j)\}$ of operator $C_l$ are integers. 

\begin{remark}\label{rem}
For a system with $n$ qubits, we have defined $l$-body phases $\Phi_l$ and operators $C_l$   for even integers  $l=0, 2, \cdots, 2\lfloor n/2\rfloor$.  As we mentioned in remark \ref{rem0} for the case of $\{C_l\}$ operators, it is sometimes useful to extend these definitions to arbitrary even integer $l\ge 0$, by defining $C_l=0$ and $\Phi_l=0$ for $l>2\lfloor n/2\rfloor$. 
\end{remark}

\subsection{Invariance under global phase transformations}

Next, we study how $l$-body phases transform under global phase transformation $V\rightarrow V'=e^{i\alpha} V$, for  $\alpha\in[0,2\pi)$.  Under this transformation $V_j$ transforms to $V'_j=e^{i\alpha} V_j$, which means $\theta_j=\text{arg}(\text{det}(V_j))$,  transforms to $\theta'_j=\theta_j+\Tr(\Pi_j) \alpha\ \ \  :\ \text{mod } 2\pi$. Then, $\Phi_l$ transforms to 
\be
\Phi'_l=\sum_j c_l(j)\  \theta_j+\sum_j c_l(j)  \Tr(\Pi_j) \alpha=\Phi_l+\Tr(C_l) \alpha=\Phi_l+2^n  \alpha\ \delta_{l,0} \ \ :\ \text{mod } 2\pi \ ,
\ee
where we have used the orthogonality relations $\Tr(C_lC_{l'})=\delta_{l,l'}\Tr(C_l^2)$, and $C_0=\mathbb{I}$. This means that $\Phi_0$ transforms to $\Phi'_0=\Phi_0+2^n  \alpha \ :\ \text{mod } 2\pi$, whereas for $l>0$,  $\Phi_l$ remains invariant. It follows that for $l>0$,  $l$-body phases do not depend on the global phase of the unitary $V$, and therefore they are physically observable. On the other hand, since $\Phi_0$ depends non-trivially on the global phase, it is not physically observable. Similarly, phases $\{\theta_j\}$ are not physically observable.

\subsection{Relation between two coordinate systems}
As mentioned before, the transformation $\{\theta_j\}\rightarrow \{\Phi_l\}$ can be interpreted as a change of the coordinate system on the $(\lfloor{{n}/{2}}\rfloor+1)-$torus defined by phases $\theta_j =\text{arg}(\text{det}(V_j))\ \text{: mod} \ 2\pi$ for $j=j_\text{min},\cdots, j_\text{max}$. In the following, we further study this relation and show how $\{\theta_j\}$ can be recovered from $l$-body phases $\{\Phi_l\}$.

Recall that $\{\Pi_j\}$ and $\{C_l\}$ are both bases for $\mathcal{C}$, the space of rotationally and permutationally-invariant operators.   Using the orthogonality relation
\be
\Tr(C_l C_{l'})=\delta_{l,l'} \Tr(C_l^2)\ ,
\ee
we find that
\be\label{ary2}
\Pi_j=\sum_{l}   \frac{\Tr(C_{l} \Pi_j)}{\Tr(C_l^2)} C_l  =\sum_{l}   \Tr(D_{l} \Pi_j)\ C_l\ ,
\ee
where we have defined 
\be
D_l\equiv  \frac{C_l}{\Tr(C_l^2)}= \frac{1}{(l/2)!\times \Tr(C_l^2)} \hspace{1mm}\sum_{i_1\neq \cdots \neq i_{l}} \hspace{-3mm} R_{i_1,i_2}\cdots\cdot R_{i_{l-1},i_{l}}  =\frac{1}{\Tr(C_l^2)}\sum_{j=j_\text{min}}^{n/2} c_{l}(j) \ \Pi_j\ .
\ee
Then, Eq.(\ref{ary1}) implies
\be\label{tnt}
\theta_j\equiv \text{arg}(\text{det}(V_j))
= -\sum_{l}   \Tr(D_{l} \Pi_j)\int_0^T dt \ \Tr(H(t) C_l)= \sum_{l}   \Tr(D_{l} \Pi_j) \times ( \Phi_l+2\pi r_l)  \ \ \ \ \ \ \text{: mod} \ 2\pi\ ,
\ee
where $\{r_l\}$ is an unspecified set of integers, and we have used Eq.(\ref{def-App}) in the last step. Therefore, up to these unspecified integers, $\{\theta_j\}$ are specified by $l$-body phases. It is worth noting that since eigenvalues of $C_l$, i.e. $c_l(j): j=j_\text{min},\cdots j_\text{max}$ are all integers, then 
\be
\Tr(D_l\Pi_j)= \frac{\Tr(\Pi_j)c_l(j)}{\Tr(C_l^2)}=\frac{ \Tr(\Pi_j) c_l(j)}{\sum_{j'} \Tr(\Pi_{j'}) c^2_l(j')}\le 1\ 
\ee
is generally less than one, which means  $\{\theta_j\}$ are not uniquely specified by $\{\Phi_l\}$. Furthermore,   because  $\Tr(D_l\Pi_j)$ is a rational number, for any given values of $\{\Phi_l\}$, there are only finitely  
many $\{\theta_j\}$  satisfying  the above equation.

\subsection{$l$-body phases of compositions of symmetric unitaries}

It turns out that $l$-body phases transform nicely under composition of symmetric unitaries. For instance, it is straightforward to see that for each even integer $l=0, 2, \cdots $, the $l$-body phase $\Phi_l$ defines a homomorphism from the group 
 of symmetric unitaries $\mathcal{V}_n$ to the group U(1) (i.e., it is a 1D representation). In particular, if $\Phi_l^{(1)}$ and $\Phi_l^{(2)}$ are the $l$-body phases of $V_1$ and $V_2$, and $\Phi_l$ is the $l$-body phase of $V_2 V_1 $, then   
  \be
\Phi_l=\Phi^{(2)}_l+\Phi^{(1)}_l\ \ \ \ \ \ : \ \ \text{mod } 2\pi\ .
\ee
For later applications, in the following  we highlight another useful property of $l$-body phases under compositions of unitaries.  While for a system with $n$ qubits, we originally defined 
$C_l$ and $\Phi_l$ for even integers $l=0, 2, \cdots,  2\lfloor n/2\rfloor $, in this proposition it is convenient to assume they are defined for all even integers $l\ge 0$, and their value is zero for $l> 2\lfloor n/2\rfloor$  (See remark \ref{rem}).

\begin{proposition}\label{propnew}
Let $V_A$ and $V_B$ be rotationally-invariant unitaries on systems $A$ and $B$ with $n_A$ and $n_B$ qubits, and  $\Phi^{(A)}_l$ and   $\Phi^{(B)}_l$ be their corresponding $l$-body phases, respectively.  Let  $\Phi^{(AB)}_l$ be the $l$-body phase of  unitary $V_A\otimes V_B$  on $n_A+n_B$ qubits. Then, for all even integer $l\ge0$, it holds that 
\be
\Phi^{(AB)}_l=2^{n_B}\times \Phi^{(A)}_l+2^{n_A}\times \Phi^{(B)}_l\ \ \ \ \ \ : \ \ \text{mod } 2\pi\ .
\ee
\end{proposition}
\begin{proof}
Let  $C^{(AB)}_{l}$  be the $C_l$ operator in Eq.(\ref{hghgh1}) defined for the joint system with $n=n_A+n_B$ qubits, and similarly $C_l^{(A)}$ and $C_l^{(B)}$ be the $C_l$ operator defined for systems $A$ and $B$, respectively.
 Then, as we saw in Eq.(\ref{partial0}),  
\be\label{partial2}
\Tr_A(C^{(AB)}_l)= 2^{n_A} \ C^{(B)}_l\ \ \ \ , \  \text{and}\ \ \ \ \   \Tr_B(C^{(AB)}_l)= 2^{n_B}\  C^{(A)}_l\ .
\ee
Suppose unitary $V_A$ is realized by a symmetric Hamiltonian $H_A(t)$ in the time interval $t\in[0, T]$. Without loss of generality, we can also assume  $V_B$ is realized by a symmetric Hamiltonian $H_B(t)$ in the same time interval. Then, unitary $V_A\otimes V_B$  is realized under Hamiltonian 
\be
H_{AB}=H_A \otimes \mathbb{I}_B + \mathbb{I}_A \otimes   H_B\ 
\ee
 in the same time interval, where $ \mathbb{I}_A$ and $ \mathbb{I}_B$ are the identity operators on systems $A$ and $B$, respectively.  This means that for all even integers  $l\ge 0$, it holds that
\bes
\begin{align}
\Phi^{(AB)}_l&= -\int_0^T dt \ \Tr(H_{AB}(t) C^{(AB)}_l)
\\ &= -\int_0^T dt \ \Tr\Big([H_A(t) \otimes \mathbb{I}_B + \mathbb{I}_A \otimes   H_B(t)] C^{(AB)}_l\Big)\\ &= -\int_0^T dt \ \Tr\big([H_A(t) \otimes \mathbb{I}_B ] C^{(AB)}_l\big) -\int_0^T dt \ \Tr\big([\mathbb{I}_A \otimes  H_B(t) ] C^{(AB)}_l\big) \\ &= -\int_0^T dt \ \Tr\big(H_A(t) \Tr_B(C^{(AB)}_l)\big) -\int_0^T dt \ \Tr\big(H_B(t) \Tr_A(C^{(AB)}_l)\big)\\ &=   -2^{n_B}\times \int_0^T dt \ \Tr\big(H_A(t) C^{(A)}_l\big) -2^{n_A}\times \int_0^T dt \ \Tr\big(H_B(t) C^{(B)}_l)\big) \\ &=2^{n_B}\times \Phi^{(A)}_l+2^{n_A}\times \Phi^{(B)}_l 
\ \ \ \ \ \ \ \ \ \ \ \ \ \ \ \ \ \ \ \ \ \ \ \ \ \ \ \ \ \ \ \ \ \ \ \  \text{: mod} \ 2\pi\ ,
\end{align}
\ees
where to get the fifth line we have used Eq.(\ref{partial2}). This proves the proposition. \end{proof}

\subsection{$l$-body phases and constraints on realizable unitaries}\label{App:const}

Recall that by definition any unitary  $V\in \mathcal{V}_k$  is realizable using $k$-local rotationally-invariant unitaries, or, equivalently,  using Hamiltonians that can be decomposed as a sum of $k$-local rotationally-invariant terms. Since operator $C_l$ is orthogonal to all $k$-local operators with $k<l$, using Eq.(\ref{def-App}) we conclude that if  unitary $V$ is realizable using $k$-local rotationally-invariant unitaries, then its  $l$-body phase $\Phi_l=0$, for all $l>k$. In the following we establish a converse statement to this result.

Let $G_0$ be the finite Abelian group 
\be\label{DefG0}
G_0\equiv \Big\langle\exp (i 2\pi D_l), \exp(i2\pi \frac{\Pi_j}{\Tr(\Pi_j)}) \ : l=0,\cdots,\lfloor{\frac{n}{2}}\rfloor ; j=j_\text{min},\cdots , n/2 \Big\rangle\ , 
\ee
 i.e., the group generated by
 unitaries $\{\exp (i 2\pi D_l)\}$ and unitaries $\{\exp(i\frac{2\pi}{\Tr(\Pi_j)} \Pi_j) \}$. The fact that these unitaries all commute with each other and the eigenvalues of $\{D_l\}$ and $\{\Pi_j\}$ are rational numbers implies that this group is finite. 
 Note that all elements of this group are in the form $\sum_j \exp(i \alpha_j) \Pi_j $, for phases $\alpha_j\in[0,2\pi)$.  
  Recall that $\mathcal{SV}_{k}\subset \mathcal{V}_{k}$ is the subgroup of rotationally-invariant unitaries $Y\cong \bigoplus_j \mathbb{I}_{2j+1}\otimes y_j$ that satisfies the additional constraint $\det(y_j)=1$ for all angular momenta $j$.  Then, according to theorem \ref{Thm1},  $\mathcal{SV}_{2}=\mathcal{SV}_{n}$.   Using this result, we show that 

\begin{theorem}\label{Thm:Phi}
Any $n$-qubit rotationally-invariant unitary $V$ with $l$-body phases $\{\Phi_l\}$ has a decomposition as
\be
V=  \tilde{V}  V_0 \prod_l \exp(i \Phi_l D_l)\ ,   
\ee 
where $V_0$ is in the finite group $G_0$ in Eq.(\ref{DefG0})  and $\tilde{V}\in\mathcal{SV}_2$. In particular, $\tilde{V}$ can be realized using the exchange interaction. 
\end{theorem}
It is worth noting that the three unitaries $\tilde{V}$, $V_0$ and $ \prod_l \exp(i \Phi_l D_l)$ in the above decomposition all commute with each other other. Since operator $D_l$ is a sum of $l$-local rotationally-invariant operators,  the unitary $\exp(i \Phi_l D_l)$
can be realized by $l$-local rotationally-invariant unitaries. We conclude that

\begin{corollary}
Suppose for $l>k$, the $l$-body phases of rotationally-invariant unitary $V$ are zero. Then, up to a unitary in the fixed  finite group $G_0$   in Eq.(\ref{DefG0}),  $V$ can be realized using $k$-local rotationally-invariant unitaries.
\end{corollary}

\begin{proof} (Theorem \ref{Thm:Phi})
We show that there exists a unitary $V_0\in G_0$ such that 
\be
 \tilde{V}=V^\dag_0  V  \prod_l \exp(-i \Phi_l D_l)\ ,
\ee
is in $\mathcal{SV}_2$.  Let 
\be
\tilde{V} \cong \bigoplus_j \tilde{V}_j \cong \bigoplus_j  (\mathbb{I}_{2j+1}\otimes \tilde{v}_j) \ 
\ee
be the decomposition of $\tilde{V}$ in the Schur basis. By virtue of theorem \ref{Thm1}, $\mathcal{SV}_2=\mathcal{SV}_n$ and therefore  to show $\tilde{V} \in\mathcal{SV}_2$ it suffices to show that 
\be
\text{det}(\tilde{v}_j)=1\ :\ \ \  j=j_\text{min},\cdots , n/2\ . 
\ee

Define
\begin{align}
U&\equiv \sum_j \exp\big(-i \frac{\theta_j}{\Tr(\Pi_j)} \big) \Pi_j\ ,
\end{align}
where $\theta_j=\text{arg}(\text{det}(V_j)) : \text{mod} 2\pi$ (The following arguments hold for any choice of $\theta_j$ satisfying this equation).   Let 
\be
W\equiv UV=VU\ ,
\ee
 and consider the decomposition 
\be
W\cong \bigoplus_j W_j \cong \bigoplus_j  (\mathbb{I}_{2j+1}\otimes w_j)\ ,
\ee
where 
\be
W_j=\exp\big(-i \frac{\theta_j}{\Tr(\Pi_j)} \big) V_j \ .
\ee
Then, 
\be
\text{det}(W_j)=\text{det}(V_j)\times \exp\big(-i \theta_j \big)=1\ .
\ee
Since $W_j=\mathbb{I}_{2j+1}\otimes w_j$, we have $\text{det}(W_j)=(\text{det}(w_j))^{2j+1}$. This means  
\be
\text{det}(w_j)=\exp(i \frac{s 2\pi}{2j+1})\ ,
\ee
where $s\in \{0,\cdots, 2j\}$. Recall that $w_j$ is a unitary transformation on the $m(n,j)$-dimensional multiplicity subsystem  of angular momentum $j$.

In conclusion, if 
$\tilde{V}\cong \bigoplus_j  (\mathbb{I}_{2j+1}\otimes \tilde{v}_j)$ is defined such that  
\be
\tilde{v}_j\equiv \exp(-i \frac{s2\pi}{(2j+1) m(n,j) }) w_j= \exp(-i \frac{s2\pi}{\Tr(\Pi_j)}) w_j\ ,
\ee
then it satisfies  the desired property
\be
\text{det}(\tilde{v}_j)=\exp(-i \frac{s2\pi}{2j+1})\times  \text{det}(w_j)=1\ ,
\ee
which implies $\tilde{V}$ belongs to $\mathcal{SV}_n=\mathcal{SV}_2$.   Therefore, in the following we define 
\begin{align}\label{ere13}
\tilde{V}\cong \bigoplus_j  (\mathbb{I}_{2j+1}\otimes \tilde{v}_j) &= \Big( \sum_j \exp(-i  \frac{s2\pi}{\Tr(\Pi_j)} ) \Pi_j\Big)  W\\ &=\Big( \sum_j \exp(-i  \frac{s2\pi}{\Tr(\Pi_j)} ) \Pi_j\Big) U V \ .
\end{align}

Next, we find $U=\sum_j \exp\big(-i \frac{\theta_j}{\Tr(\Pi_j)} \big) \Pi_j$.  Recall that  $\theta_j$ can be any  phase satisfying $\theta_j=\text{arg}(\text{det}(V_j)) : \text{mod} 2\pi$. In particular, using Eq.(\ref{tnt}) we can choose  $\theta_j=\sum_{l}   \Tr(D_{l} \Pi_j) \times ( \Phi_l+2\pi r_l)$,  where $r_l$ is a set of unspecified integers.  This implies
\bes
\begin{align}
U=\sum_j \exp\big(-i \frac{\theta_j}{\Tr(\Pi_j)} \big) \Pi_j\ &=\sum_j   \exp\Big(-i \sum_{l}   \frac{\Tr(D_{l} \Pi_j) \times ( \Phi_l+2\pi r_l)}{\Tr(\Pi_j)}  \Big) \Pi_j  \\ &= \prod_l   \exp\Big(-i  ( \Phi_l+2\pi r_l)  \sum_j  \frac{\Tr(D_{l} \Pi_j) \Pi_j}{\Tr(\Pi_j)}   \Big)\\ &= \prod_l   \exp\Big(-i  ( \Phi_l+2\pi r_l)  D_l  \Big) \\ &= \prod_l   \exp(-i \Phi_l D_l)   \prod_l   \exp(-i 2\pi r_l D_l) \ .
\end{align}
\ees
Combining this with Eq.(\ref{ere13}),  we conclude that
\bes
\begin{align}
V&=\tilde{V}  U^\dag  \Big[ \sum_j \exp(i  \frac{s2\pi}{\Tr(\Pi_j)} ) \Pi_j\Big] \\ &= \tilde{V} \Big[\prod_l   \exp(i \Phi_l D_l)  \Big]   \Big[\prod_l  \exp(i 2\pi r_l D_l)\Big]  \Big[ \sum_j \exp(i  \frac{s2\pi}{\Tr(\Pi_j)} ) \Pi_j\Big]\\ &= \tilde{V} \Big[\prod_l   \exp(i \Phi_l D_l)  \Big] V_0 \ , 
\end{align}
\ees
where $V_0\equiv \Big[\prod_l  \exp(i 2\pi r_l D_l)\Big]  \Big[ \sum_j \exp(i  \frac{s2\pi}{\Tr(\Pi_j)} ) \Pi_j\Big]$ is in the finite group $G_0$, and $\tilde{V}\in\mathcal{SV}_2$. This completes the proof. 
\end{proof}

\newpage

\section{Multi-qubit swap Hamiltonian has overlap with all $C_l$ operators}\label{App:Example}

In this appendix we consider the multi-qubit swap Hamiltonian $S_{AB}$ and  show that it is impossible to implement $S_{AB}$ with $k$-local rotationally-invariant Hamiltonians with $k < n$. In the next section, we explicitly show that it is possible to implement $S_{AB}$ with only 2-local exchange interactions, provided that one can use ancillary qubits.

Let us divide a system with $n=2r$ qubits into two equal-sized subsystems $A$ and $B$. Then, the  multi-qubit swap operator $S_{AB}$ is defined as the swap operator satisfying $S_{AB}(|\eta\>_A|\gamma\>_B)=|\gamma\>_A|\eta\>_B$,
for any pair of $r$-qubit states states $|\eta\>$ and $|\gamma\>$.  For concreteness, we assume $A$ and $B$ are the set of odd and even qubits, respectively.  Then, the  multi-qubit swap operator $S_{AB}$ can be written as
\begin{align}
  S_{AB} = \P_{12}\cdots \P_{n-1,n}\ ,
\end{align}
where $\P_{ij}$ is the 2-qubit swap operator that only exchanges the states of qubits $i$ and $j$, and leaves other qubits unchanged.

In the following we show that
$\Tr(S_{AB} C_l) \neq 0$ for all  $l=0,2 ,\cdots, n=2r$. Together with  \cref{newThm}, this implies that Hamiltonian  $S_{AB}$ cannot be implemented using $k$-local rotationally-invariant Hamiltonians when $k<n$.

To make the result more general, let us consider the operator $\P_{12}\cdots \P_{m-1,m}$. Using the relation $\P_{ij} = R_{ij} + \frac{\mathbb{I}}{2}$, we obtain the expansion 
\begin{align}
  \P_{12}\cdots \P_{m-1,m} = \Big(R_{12} + \frac{\mathbb{I}}{2}\Big) \cdots \Big(R_{m-1,m} + \frac{\mathbb{I}}{2}\Big) = \sum_{s=0}^{m} 2^{(s-m)/2} \sum_{1 \leq i_1 < \cdots < i_{s/2} \leq l/2} R_{2i_1-1,2i_1} \cdots R_{2i_{s/2}-1,2i_{s/2}}\ ,
\end{align}
where the summation is over even integer $s$. Recall the definition
\be\nonumber
C_l\equiv \frac{1}{(l/2)!}  \sum_{i_1\neq i_2\neq\cdots \neq i_{l}} R_{i_1,i_2}\cdots\cdot R_{i_{l-1},i_{l}} \ \   \ \   \ \   \ \  :  l=2, 4 \cdots, 2\lfloor \frac{n}{2}\rfloor\ .
\ee
 This implies 
\begin{align}
   \Tr( \P_{12}\cdots \P_{m-1,m} C_{l}) &= \sum_{s=0}^m 2^{(s-m)/2} \sum_{1 \leq i_1 < \cdots < i_{s/2} \leq m/2} \Tr(R_{2i_1-1,2i_1} \cdots R_{2i_{s/2}-1,2i_{s/2}} C_l ) \nonumber\\
  &= \sum_{s} 2^{(s-m)/2} \binom{m/2}{s/2} \Tr(R_{12} \cdots R_{s-1,s} C_{l} ) \nonumber\\
  &= \sum_{s} 2^{(s-m)/2} \binom{m/2}{s/2} \frac{(n-s)!}{n!} \sum_{i_1\neq \cdots \neq i_s}\Tr(R_{i_1i_2}\cdots R_{i_{s-1}i_{s}} C_{l} ) \nonumber\\
  &= \sum_{s} 2^{(s-m)/2} \binom{m/2}{s/2} \frac{(n-s)!}{n!} (s/2)! \Tr( C_s C_{l} ) \nonumber\\
  &= 2^{(l-m)/2} \frac{(n-l)!(m/2)!}{n!(m/2-l/2)!} \Tr( C_l^2 )\ ,\label{eq:TrPC}
\end{align}
where the summations are over even integers  $s$, and have used the fact that $C_l$ is permutationally invariant to obtain the second and third lines. The fourth line follows from the definition of $C_s$ operator and to get the last line we have used the orthogonality relations $\Tr(C_lC_s)=\Tr(C_l^2) \delta_{l,s}$.  In summary, we have
\begin{align}
  \Tr(S_{AB} C_l) &= 2^{(l-n)/2} \frac{(n-l)!(n/2)!}{n!(n/2-l/2)!} \Tr(C_l^2) > 0.
\end{align}
for all $l=0,2 ,\cdots, n=2r$. This implies that Hamiltonian  $S_{AB}$ cannot be implemented using $k$-local rotationally-invariant Hamiltonians when $k<n$.

\subsection*{A new basis for the space of permutationally and rotationally invariant operators}
The above calculation motivates us to define a new basis for the space of  permutationally and rotationally invariant 
operators (This basis is not used in the rest of the paper). 
Recall that 
we have already found two other bases for this space,  namely $\{\Pi_j\}$ and $\{C_l\}$. For $m=2, 4 \cdots, 2\lfloor \frac{n}{2}\rfloor$ define
\begin{align}
  B_m &\equiv \frac{1}{(m/2)!} \sum_{i_1\neq \cdots \neq i_m} \P_{i_1i_2}\cdots \P_{i_{m-1}i_{m}} = \frac{1}{(m/2)!} \sum_{i_1\neq \cdots \neq i_m} \Big(R_{i_1i_2} + \frac{\mathbb{I}}{2}\Big) \cdots \Big(R_{i_{m-1}i_m} + \frac{\mathbb{I}}{2}\Big)\ .
\end{align}

This definition is analogous to the definition of $C_l$, except we have replaced  $R_{ij}$ with $\P_{ij}=R_{ij}+\mathbb{I}/2$. Next, we determine the relation between this basis and $\{C_l\}$ and $\Pi_j$ bases. Since  
\begin{align}
  B_m=\sum_{l} \frac{\Tr(B_m C_l)}{\Tr(C_l^2)} C_l\ ,
\end{align}
it suffices to calculate the ratio $\frac{\Tr(B_m C_l)}{\Tr(C_l^2)}$. Using the fact that $C_l$ is permutationally invariant, we have
\begin{align}
  \Tr(B_m C_l) = \frac{1}{(m/2)!} \sum_{i_1\neq \cdots \neq i_m} \Tr( \P_{i_1i_2}\cdots \P_{i_{m-1}i_{m}} C_l) = \frac{1}{(m/2)!} \frac{n!}{(n-m)!} \Tr( \P_{12}\cdots \P_{m-1,m} C_l)\ . 
\end{align}
We have already calculated this trace in \cref{eq:TrPC}, and therefore we have
\begin{align}
  \Tr(B_m C_l) = \frac{1}{(m/2)!} \frac{n!}{(n-m)!} 2^{(l-m)/2} \frac{(n-l)!(m/2)!}{n!(m/2-l/2)!} \Tr( C_l^2 ) = 2^{(l-m)/2} \frac{(n-l)!}{(n-m)!(m/2-l/2)!} \Tr( C_l^2 )\ . 
\end{align}
Therefore
\begin{align}
   B_m = \sum_{l=0}^m 2^{(l-m)/2} \frac{(n-l)!}{(n-m)!(m/2-l/2)!} C_l\ .
\end{align}
Using the relation $C_l=\sum_j c_l(j) \Pi_j$, we also obtain
\be
 B_m =\sum_j b_m(j) \Pi_j\ ,
\ee
where
\be
b_m(j) = \sum_{l=0}^m 2^{(l-m)/2} \frac{(n-l)!}{(n-m)!(m/2-l/2)!}  c_l(j)\ .
\ee

\newpage

\section{Achieving universality with the exchange interactions and two ancilla qubits}\label{App:univ}
In the main paper, we described  a recursive approach for implementing a general rotationally-invariant unitary using the exchange interaction and two ancilla qubits. Here, we present further details and a more formal statement of the result (See theorem \ref{Thm:ancilla} below). Then, in Sec.\ref{App:exp} we  present an independent proof of this result for a special family of rotationally-invariant Hamiltonians. The latter approach yields an explicit construction of the desired Hamiltonian that couples the system to ancillae, in terms of the exchange interactions and their nested commutators.  

\subsection{Implementing general rotationally-invariant Hamiltonians with the exchange interactions}  

In the following, we consider a system with $n$ qubits and a pair of ancilla qubits denoted by $a$ and $b$, with the total Hilbert space  $(\mathbb{C}^2)^{\otimes (n+2)}$.   We usually assume the ancilla qubits are initially prepared in state $|00\rangle_{ab}$. But, as we argued in the main paper, due to the rotational symmetry of the process,  the same construction works if qubits are in any (possibly mixed)  state whose support is restricted to the triplet subspace. Also, a slightly modified version of this scheme works if the ancilla qubits are initially in the singlet state $(|01\rangle-|10\rangle)/\sqrt{2}$.  

  We show that for any rotationally-invariant Hamiltonian $H$
 on $n$ qubits, there exists a corresponding  rotationally-invariant Hamiltonian $\widetilde{H}$ on $n+2$ qubits, whose effect on the main $n$-qubit system is equivalent to   Hamiltonian $H$, and, furthermore, it can be implemented using the exchange interaction (See theorem \ref{Thm:ancilla} for the formal statement). The latter property means that $i \widetilde{H}$ is in the real Lie algebra  generated by operators 
\be\label{gen1}
i \mathbb{I}\otimes I_{ab}\ ,\ i R_{rs}\ ,\   i R_{ab}\ ,\  i R_{ra}\ ,\ i R_{rb} : \ \ \  1 \le r< s \le n\ \  ,
\ee
that is
\be
i \widetilde{H}_l\in \widetilde{\mathfrak{v}}_2\equiv \mathfrak{alg}_\mathbb{R}\big\{i \mathbb{I}\otimes I_{ab},  i R_{rs},  i R_{ab}, i R_{ra} , i R_{rb} : \ \ \  1 \le r< s \le n\ \big\} \ ,
\ee
where $\mathbb{I}$ and   $I_{ab}$, respectively, denote the identity operators on $n$ qubits in the main system and   on the ancilla qubits $a$ and $b$.

 Note that the exchange interactions generate single-qubit swaps (transpositions): using the fact that
\be
\textbf{P}_{rs}=R_{rs}+\frac{\mathbb{I}}{{2}}\ ,
\ee
one can easily see 
\be
\exp(i \frac{\pi}{2}R_{rs} )=\exp(-i\frac{\pi}{4}) \exp(i \frac{\pi}{2}\textbf{P}_{rs})=\exp(i\frac{\pi}{4})\textbf{P}_{rs}\ .
\ee
Combinations of swaps generate all permutations and allow arbitrary reordering of qubits. Therefore, to establish the above result  direct exchange interactions between all pairs of qubits are not required. That is, the Lie algebra generated by operators in Eq.(\ref{gen1}) is equal to the Lie algebra generated by operators 
\be\label{gen3}
i \mathbb{I}\otimes I_{ab}\ ,  i R_{r, r+1}\ , i R_{na}\ ,  i R_{ab}\ : \ \ \  1 \le r < n\ ,
\ee
which corresponds to the nearest-neighbor interactions between $n+2$ qubits on a chain. 

To prove the existence of Hamiltonian $\widetilde{H}$ with the above desired properties,  we apply theorem \ref{Thm1} to the $(n+2)$-qubit system.  Let $\widetilde{\Pi}_j$ be the projector to the subspace with the total angular momentum $j$ of $n+2$ qubits, including the ancilla qubits $a$ and $b$.
More precisely, $\widetilde{\Pi}_j$ is the projector to the eigen-subspace of $\widetilde{J}^2$ with eigenvalue $j(j+1)$, where
\be
\widetilde{J}^2=\frac{1}{4}\sum_{r=x,y,z}\Big( \sigma^{(r)}_a+\sigma^{(r)}_b+\sum_{i=1}^n \sigma_i^{(r)}\Big)^2\ .
\ee

 Then, according to theorem \ref{Thm1}, any rotationally-invariant Hamiltonian $\widetilde{F}$ on    
$n+2$ qubits can be implemented using the exchange interactions, provided that it satisfies the condition
\be
\Tr(\widetilde{\Pi}_j \widetilde{F})=0 \ \ \ \ :\  j=j_\text{min},\cdots , \frac{n}{2}+1\ ,
\ee 
where $j_\text{min}=0 , 1/2$, for even and odd $n$, respectively.

We start by finding the Hamiltonian $\widetilde{Q}_l$ that corresponds to $n$-qubit Hamiltonian $Q_l=R_{12}\cdots R_{l-1,l}$.   For $l=2,\cdots, 2\lfloor n/2\rfloor$  the following operators are defined on  $(n+2)$ qubits, including  ancilla qubits $a$ and $b$:  
\begin{align}
\widetilde{Q}_l&\equiv \frac{1}{2^{\frac{l}{2}}} \Big[ 2 Q_2\otimes I_{ab}+\sum_{s=4}^{l} {2^{s/2}} (Q_{s}\otimes I_{ab}-Q_{s-2} \otimes R_{ab})  \Big] \\ &=(Q_l\otimes I_{ab}-Q_{l-2}\otimes  R_{ab})+ \frac{1}{2}(Q_{l-2}\otimes I_{ab}-Q_{l-4}\otimes  R_{a,b})+\cdots\cdots  + \frac{1}{2^{\frac{l}{2}-1}}Q_{2}\otimes I_{ab}\  , 
 \end{align}
where $I_{ab}$ denotes the identity operator on the 4D Hilbert space of the ancilla qubits.  We claim that Hamiltonian $\widetilde{Q}_l$ can be implemented using the exchange interaction, that is  $i \widetilde{Q}_l\in \widetilde{\mathfrak{v}}_2$. 
To see this note for each term 
$Q_{s}\otimes I_{ab}-Q_{s-2} \otimes R_{ab}$,  it holds that 
\be
\Tr(\widetilde{\Pi}_j [Q_{s}\otimes I_{ab}-Q_{s-2} \otimes R_{ab}])=0 \ \ \ \ :\  \ \  j=j_\text{min},\cdots , \frac{n}{2}+1\ .
\ee
This follows from the fact that $\widetilde{\Pi}_j$ is permutationally-invariant and 
operators $Q_{s}\otimes I_{ab}$ and $Q_{s-2} \otimes R_{ab}$ are equal to each other, up to a permutation that exchanges qubits $s-1$ and $s$ with qubits $a$ and $b$. Then, applying  theorem \ref{Thm1}, we conclude that $i\sum_{s=4}^{l} {2^{s/2}} (Q_{s}\otimes I_{ab}-Q_{s-2} \otimes R_{ab})$ is in the Lie algebra $\widetilde{\mathfrak{v}}_2 $, which in turn implies $i \widetilde{Q}_l\in \widetilde{\mathfrak{v}}_2$.

Next, we note that operator $\widetilde{Q}_l$ can be rewritten as 
\be\label{ert}
\widetilde{Q}_l=Q_l\otimes I_{ab}+ \frac{1}{2^{\frac{l}{2}-1}} \sum_{r=2}^{l-2} {2^{r/2}}\  Q_{r}\otimes (\frac{I_{ab}}{2}- R_{ab}) \ . 
\ee
This implies that for any state $|\psi\rangle\in(\mathbb{C}^2)^{\otimes n}$,  it holds that 
\be
\widetilde{Q}_l (|\psi\rangle\otimes |00\rangle_{ab})=Q_l|\psi\rangle \otimes |00\rangle_{ab}+ \frac{1}{2^{\frac{l}{2}-1}} \sum_{r=2}^{l-2} {2^{r/2}}\  Q_{r}|\psi\rangle\otimes(\frac{I_{ab}}{2}- R_{ab}) |00\rangle_{ab} =Q_l|\psi\rangle \otimes |00\rangle_{ab}\ ,
\ee
where we have use the fact that $\frac{I_{ab}}{2}- R_{ab}$  is proportional to the projector to the singlet subspace of ancilla qubits, which is orthogonal to $|00\rangle_{ab}$.  We conclude that, under Hamiltonian $\widetilde{Q}_l$, the evolution of the $n$-qubit system is described by Hamiltonian $Q_l$. That is, using 2 ancilla qubits we can implement Hamiltonian $Q_l$ via the exchange interaction alone.

Next, define $\mathcal{P}$ to be the linear map on the space of operators that projects any operator to the space of operators that are invariant under permutations of $n$ qubits in the system. In particular,
\be
\mathcal{P}(R_{i_1,i_2}\cdots R_{i_{l-1},i_l}R_{ab})=\frac{1}{n!}\sum_{\sigma\in\mathbb{S}_n}  R_{\sigma(i_1),\sigma(i_2)}\cdots R_{\sigma(i_{l-1}),\sigma(i_l)} R_{ab}\ ,
\ee 
where the summation is over $\mathbb{S}_n$, the group of permutations on $n$ objects. Note that this map acts trivially on the ancilla qubits $a$ and $b$. Using this notation we have
\begin{align}
C_l&\equiv \frac{1}{(l/2)!}  \sum_{i_1\neq i_2\neq\cdots \neq i_{l}} R_{i_1,i_2}\cdots\cdot R_{i_{l-1},i_{l}}\\ &=    \frac{1}{(l/2)!}\times \frac{1}{(n-l)!}   \sum_{\sigma\in\mathbb{S}_n}  R_{\sigma(1),\sigma(2)}\cdots \cdots  R_{\sigma({l-1}),\sigma(l)} \\ &= \frac{n!}{(n-l)!\times (l/2)!}  \mathcal{P}(Q_l)\ .
\end{align}
Then, for even integer $l=2,\cdots , 2\lfloor n/2\rfloor$ define 
\begin{align}
\widetilde{E}_l&\equiv \frac{n!}{(n-l)!\times (l/2)!} \mathcal{P}(\widetilde{Q}_l)\\ &=C_l\otimes I_{ab}+  2\sum_{r=2}^{l-2} 
\frac{{2^{r/2}}\times  (n-r)! \times (r/2)! }{2^{l/2} \times (n-l)!\times (l/2)!}\  C_{r}\otimes (\frac{I_{ab}}{2}- R_{ab}) \ , 
\end{align}
where we have used Eq.(\ref{ert}). 
Since $\widetilde{E}_l$ is a linear combination of $\widetilde{Q}_l$ and its permuted versions, similar to $\widetilde{Q}_l$ it can be realized using exchange interactions. That is, $i \widetilde{E}_l\in \widetilde{\mathfrak{v}}_2$. Furthermore, for all 
$|\psi\rangle\in(\mathbb{C}^2)^{\otimes n}$,  it holds that 
\be
\widetilde{E}_l (|\psi\rangle\otimes |00\rangle_{ab})=C_l|\psi\rangle \otimes |00\rangle_{ab}\ .
\ee
It follows that in this way we can implement all Hamiltonians  $C_l: l=2,\cdots, 2\lfloor n/2\rfloor$ using  exchange interactions and two ancillary qubits.  Finally, recall that any rotationally-invariant Hamiltonian $H$ can be decomposed as  
\be
H= H_0+\frac{\Tr(H)}{2^n}\mathbb{I} + \sum_{l=2}^{2\lfloor n/2\rfloor}\frac{\Tr(C_l H)}{\Tr(C_l^2)} C_l \ ,
\ee
where $H_0$ satisfies $\Tr(H_0\Pi_j)=0$ for all $j$. Again, by applying theorem  \ref{Thm1}, we find that $H_0$
 can be implemented using  exchange interactions alone.
  Finally, note that because $H_0 \otimes I_{ab}$ and $\widetilde{E}_l: 2,\cdots, 2\lfloor n/2\rfloor$ are all traceless, they are all orthogonal to the identity operator $\mathbb{I}\otimes I_{ab}$, and therefore, they all belong to the sub-algebra generated by   $\mathfrak{alg}_\mathbb{R}\{iR_{ab}, i R_{n, a}, i R_{r, r+1}: \ 1\le r< n\}$.

  In summary, we conclude that \begin{theorem}\label{Thm:ancilla}
For any  rotationally-invariant Hamiltonian $H$ on $n$ qubits, consider Hamiltonian 
\be
\widetilde{H}=(H-\frac{\Tr(H)}{2^n}\mathbb{I})\otimes I_{ab}+ \sum_{l=2}^{2\lfloor n/2\rfloor} \frac{\Tr(H C_l)}{\Tr(C_l^2)} (\widetilde{E}_l-C_l\otimes  I_{ab})\ , 
\ee
defined on $n$ qubits and two  ancilla qubits $a$ and $b$. Then, $\widetilde{H}$ can be realized using the exchange interaction alone, i.e., 
\be
i\widetilde{H}\in \mathfrak{alg}_\mathbb{R}\{iR_{ab}, i R_{n, a}, i R_{r, r+1}: \ 1\le r< n\}\ .
\ee
 Furthermore, for all $|\psi\rangle\in(\mathbb{C}^2)^{\otimes n}$, it satisfies 
\be
\widetilde{H}_l (|\psi\rangle\otimes |00\rangle_{ab})=(H-\frac{\Tr(H)}{2^n}\mathbb{I}) |\psi\rangle \otimes |00\rangle_{ab}\ .
\ee
\end{theorem}
Note that the latter property means that for all time $t\in\mathbb{R}$, it holds that
\be
\e^{-i t \widetilde{H}} (|\psi\rangle\otimes  |00\rangle_{ab})=\e^{i\Tr(H)t/2^n} (\e^{-i t {H}} |\psi\rangle)\otimes  |00\rangle_{ab}\ ,
\ee
where $\e^{i\Tr(H)t/2^n}$ is a global phase.

\subsection{Explicit implementation of a family of rotationally-invariant Hamiltonians including the multi-qubit swap Hamiltonian}\label{App:exp}
 The above result is not constructive, i.e., it does not provide a direct decomposition of Hamiltonian $\widetilde{H}$ into the exchange interactions and its nested commutators. In the following, we address this issue for a special class of rotationally-invariant Hamiltonians, namely  those that can 
be written as a linear combination of monomials in the form $R_{12} R_{34},\cdots$, and their permuted versions.  More precisely, consider any Hamiltonian $H$ in the form  
\be\label{tra}
H=\sum_{l=2}^{2\lfloor n/2\rfloor}\sum_{i_1 \neq i_2 \neq \cdots \neq i_{l}} h_l(i_1,\cdots, i_l)\   R_{i_1,i_2}\cdots\cdot R_{i_{l-1},i_{l}} \ ,
\ee
where the summation is over even integer $l$'s and $\{h_l\}$ are arbitrary real functions. Multi-qubit swap Hamiltonian  and operators $\{C_l\}$ are examples of Hamiltonians that can be written in the above form, up to a shift by a multiple of the identity operator.

In the following, we present a recursive approach for decomposing Hamiltonian $Q_l-Q_{l-2}R_{a,b}$ to the exchange interactions and their nested commutators. Then, using this Hamiltonian and its permuted version we can realize any Hamiltonian in the form of Eq.(\ref{tra}). The argument is similar to what we have seen in the previous section: First,  using the fact that  
\be
 \big(Q_l-Q_{l-2}R_{a,b} \big) |\psi\rangle  \otimes  |00\rangle_{ab}=(Q_l-\frac{1}{2}Q_{l-2})   |\psi\rangle \otimes  |00\rangle_{ab}\  ,
\ee
we find
\be
e^{-i t (Q_l-Q_{l-2}R_{a,b})} |\psi\rangle \otimes |00\rangle_{ab}=(e^{-i t (Q_l-Q_{l-2}/2)})|\psi\rangle  \otimes  |00\rangle_{ab} \ ,
\ee
where $|\psi\rangle$ is an arbitrary  state of $n$ qubits in the system.  Therefore, in this way we can implement Hamiltonian $Q_l-Q_{l-2}/2$. To obtain Hamiltonian $Q_l$, we use the expansion 
\be\label{art56}
Q_l=(Q_l-\frac{1}{2}Q_{l-2})+  \frac{1}{2}(Q_{l-2}-\frac{1}{2} Q_{l-4})+\cdots \cdots +\frac{1}{2^{\frac{l}{2}-2}}(Q_4-\frac{1}{2}Q_2)+\frac{1}{2^{\frac{l}{2}-1}} Q_2\ ,
\ee
where $Q_2=R_{12}$.  We conclude that, if one can implement Hamiltonian $Q_l-Q_{l-2}R_{a,b}$ for general even integer $l$, then one can implement general Hamiltonian $H$ in  Eq.(\ref{tra}). 

\subsection*{Explicit implementation of $Q_l-Q_{l-2}R_{a,b}$}
Finally,  we directly  show that $Q_l-Q_{l-2}R_{a,b}$ is in the Lie algebra generated by the exchange interactions.  First, we establish this for $l=4$.   One can check the following indentities 
\begin{align}
  \big[ [[[R_{12}, R_{23}], R_{34}], R_{45}], R_{51}\big] &= 2 (-R_{12} R_{34} - R_{14} R_{23} + R_{23} R_{45} + R_{25} R_{34})\ , \\
  \big[ [R_{12}, R_{23}], R_{34} \big] &= 2 (R_{14} R_{23} - R_{13} R_{24})\ .
\end{align}
These equations imply 
\begin{align}
  4(R_{12} R_{34} - R_{23} R_{45}) &= - \big[ [[[R_{12}, R_{23}], R_{34}], R_{45}], R_{51}\big] - \big[ [R_{13}, R_{23}], R_{24} \big] + \big[ [R_{24}, R_{34}], R_{35} \big]\ .
\end{align}
Then, we arrive at the interesting identity
\begin{align}\label{eq:shiftR2-22}
  4(R_{12} R_{34} - R_{34} R_{56})  &= 4(R_{12} R_{34} - R_{23} R_{45}) + 4(R_{23} R_{45} - R_{34} R_{56}) \nonumber\\
  &= -\big[ [[[R_{12}, R_{23}], R_{34}], R_{45}], R_{51}\big] - \big[ [[[R_{23}, R_{34}], R_{45}], R_{56}], R_{62}\big]  - \big[ [R_{13}, R_{23}], R_{24} \big] + \big[ [R_{35}, R_{45}], R_{46} \big] \ .
\end{align}
Therefore, if we  we choose qubits $5$ and $6$ to be ancilla qubits $a$ and $b$, we obtain
\begin{align}\label{eq:shiftR2-22-relabel}
  4(Q_4 - R_{ab} R_{34}) &= -\big[ [[[R_{12}, R_{23}], R_{34}], R_{4a}], R_{a1}\big] - \big[ [[[R_{23}, R_{34}], R_{4a}], R_{ab}], R_{b2}\big]  - \big[ [R_{13}, R_{32}], R_{24} \big] + \big[ [R_{3a}, R_{a4}], R_{4b} \big] \ .
\end{align}
Note that up to a permutation this is equal to $Q_4-R_{12}R_{ab}=Q_4-Q_2 R_{ab}$.

Next, we use this identity recursively to obtain $Q_l-Q_{l-2} R_{ab}$.  Because of the form of Eq.(\ref{eq:shiftR2-22-relabel}), to simplify the notation it is more convenient  to construct  a  permuted version of $Q_l-Q_{l-2}R_{a,b}$, namely $Q_l - R_{ab} R_{34} T_l$, where
\begin{align}
  T_l \equiv R_{56} \cdots R_{l-1, l}\ .
\end{align}
Note that operator $Q_l - R_{ab} R_{34} T_l$ can be transformed to  $Q_l-Q_{l-2}R_{ab}$ by exchanging qubits $1 \leftrightarrow l-1$ and $2 \leftrightarrow l$.

For a general even integer $l$ suppose we have found the decomposition of $Q_l - R_{ab} R_{34} T_l$
 in terms of the nested commutators of the exchange interactions. We show that by properly inserting these terms in \cref{eq:shiftR2-22-relabel} we can obtain  $Q_{l+2} - R_{ab} R_{34} T_{l+2}$.  To see this suppose  in the right-hand side of \cref{eq:shiftR2-22-relabel},  the first commutator in each nested commutator is changed from $[R_{rs}, R_{st}]$ to
\begin{align}
  [R_{rs}, R_{st}] \ \longrightarrow \ [T_{l+2} R_{rs} - T_{l+2} R_{uv}, R_{st}] = [T_{l+2} R_{rs}, R_{st}]\ ,
\end{align}
where $r,s,t \in \set{1,2,3,4,a}$ as in \cref{eq:shiftR2-22-relabel} and $u,v \in \set{1,2,3,4,a,b}\setminus\set{r,s,t}$. The equality holds because $[T_{l+2} R_{uv}, R_{st}] = 0$. Note that, up to a permutation of qubits, $T_{l+2} R_{rs} - T_{l+2} R_{uv}$ is in the form $Q_l - R_{ab} R_{34} T_l$ and therefore by assumption we know its decomposition in terms of the nested commutators of the exchange interactions. 
 With these replacements Eq.(\ref{eq:shiftR2-22-relabel}) transforms to 
\begin{align}
  &\quad -\big[ [[[T_{l+2} R_{12}, R_{23}], R_{34}], R_{4a}], R_{a1}\big] - \big[ [[[T_{l+2} R_{23}, R_{34}], R_{4a}], R_{ab}], R_{b2}\big] - \big[ [T_{l+2} R_{13}, R_{32}], R_{24} \big] + \big[ [T_{l+2} R_{3a}, R_{a4}], R_{4b} \big] \nonumber\\
  &= T_{l+2} \big(-\big[ [[[R_{12}, R_{23}], R_{34}], R_{4a}], R_{a1}\big] - \big[ [[[R_{23}, R_{34}], R_{4a}], R_{ab}], R_{b2}\big]  - \big[ [R_{13}, R_{32}], R_{24} \big] + \big[ [R_{3a}, R_{a4}], R_{4b} \big] \big) \nonumber\\
  &= 4 T_{l+2} (Q_4 - R_{ab} R_{34}) = 4 (Q_{l+2} - R_{ab} R_{34} T_{l+2})\ .
\end{align}
In the first equality, we can factor out $T_{l+2}$ because it commutes with all the other terms. The second equality just follows from \cref{eq:shiftR2-22-relabel}.

In summary, assuming we have found the decomposition of $Q_l - R_{ab} R_{34} T_l$ in terms of the commutators of the exchange interactions, then we have also an explicit construction for operator $Q_{l+2} - R_{ab} R_{34} T_{l+2}$. Combining this with \cref{eq:shiftR2-22-relabel}, which corresponds to the special case of $l=4$, we obtain a general construction of 
\begin{align}
  Q_l - R_{ab} R_{34} T_l = (R_{12} - R_{ab}) R_{34} \cdots R_{l-1,l}.
\end{align}
 for arbitrary even $l$. 
 As mentioned before, up to a permutation of qubits, this is equivalent to $Q_l-Q_{l-2}R_{ab}$.

Combining this with the argument in Eq.(\ref{art56}), we find a recursive approach for  constructing a Hamiltonian $\widetilde{H}$ that realizes an arbitrary Hamiltonian $H$ in the form of Eq.(\ref{tra}),  such that $\widetilde{H}(|\psi\rangle|00\rangle_{ab})=(H|\psi\rangle)|00\rangle_{ab}$. \\
  
\newpage

\section{Insufficiency of a single ancilla qubit for universality}\label{App:single}
Finally, we show that two ancilla qubits are necessary to achieve universality. That is, generic rotationally-invariant unitaries cannot be implemented with local rotationally-invariant unitaries, even if one can use a single ancilla qubit.

Suppose using a single ancillary qubit $a$ in the initial state $|\eta\rangle$  we can implement rotationally-invariant unitary $V$, such that for any arbitrary initial state $|\psi\rangle\in(\mathbb{C}^2)^{\otimes n}$, it holds that
\be
\widetilde{V} (|\psi\rangle\otimes |\eta\rangle) =({V} |\psi\rangle)\otimes |\eta'\rangle \ , 
\ee 
where $|\eta'\rangle$ is the final state of the ancilla and $\widetilde{V}$ is itself rotationally-invariant, i.e., satisfies
\be
\big[\widetilde{V} , U^{\otimes (n+1)}\big]=0\ ,
 \ee
for all single-qubit unitary $U$. Note that, since by assumption the final state of the qubits in the main system is pure, the final state of the ancilla should also be pure. Furthermore, linearity of $\widetilde{V}$ implies that $ |\eta'\rangle$ cannot depend on $|\psi\rangle$. The fact that $V$ is rotationally-invariant means that the total angular momentum of $n$ qubits in the system is conserved, which in  turn implies that the angular momentum of the ancilla is also conversed. This means $|\eta'\rangle$ is equal $|\eta\rangle$, up to a global phase. Since we can always absorb this phase in the definition of $\widetilde{V}$, we can assume  $|\eta'\rangle=|\eta\rangle$. 

In summary, the most general case where one uses an ancilla qubit to perform a rotationally-invariant unitary $V$ on $n$ qubits can be formulated as 
 \be\label{lastapp}
\widetilde{V} (|\psi\rangle\otimes |\eta\rangle) =({V} |\psi\rangle)\otimes |\eta\rangle \ . 
\ee 
 Next, note that for any single-qubit unitary $U$, it holds that
  \begin{align}
\widetilde{V} \big(|\psi\rangle\otimes U |\eta\rangle\big) &=\widetilde{V} \Big(U^{\otimes n}(U^{\dag})^{\otimes n}|\psi\rangle\otimes U |\eta\rangle\Big)\\ &=U^{\otimes (n+1)}\widetilde{V} \Big((U^{\dag})^{\otimes n}|\psi\rangle\otimes |\eta\rangle\Big)\\ &=U^{\otimes (n+1)}\Big({V} (U^{\dag})^{\otimes n}|\psi\rangle\otimes |\eta\rangle\Big)\\ &={V}|\psi\rangle\otimes U |\eta\rangle\ , 
\end{align}
where the second and fourth lines follow from the rotational symmetry of $\widetilde{V} $ and $V$, and the third line follows from the assumption that Eq.(\ref{lastapp}) holds for all states  $|\psi\rangle\in(\mathbb{C}^2)^{\otimes n}$.  States in the form $\{|\psi\rangle\otimes U|\eta\rangle\}$, for different $|\psi\rangle\in(\mathbb{C}^2)^{\otimes n}$ and single-qubit unitaries $U$, span the total Hilbert space $(\mathbb{C}^2)^{\otimes (n+1)}$ corresponding to $n+1$ qubits. Since the above identity holds for all such states, we conclude that $ \widetilde{V}$ acts as the identity operator on the ancilla qubits, that is 
 \be
 \widetilde{V}=V\otimes I_a \ , 
 \ee
 where $I_a$   is the identity operator on the Hilbert space of the ancilla qubit.
 Therefore,  in principle, the unitary $\widetilde{V}$ can be realized without any interactions between the system and ancilla.   However, this does not necessarily mean that the ancilla is useless.  Indeed, 
it is conceivable that the ancilla has interacted with the system, but at the end it has become decoupled from the system again. 
  In general, even for  Hamiltonians $\widetilde{H}$ that involve interactions between ancilla qubit $a$ and the main system, at certain moments of time $t$ the overall unitary  
transformation $\exp(-i t \widetilde{H})$ can be in the form $V\otimes I_a$.  Now, here is the key point:  a priori,  it is not clear that unitaries that can be realized in this way can also be obtained  using realizable Hamiltonians that do not contain system-ancilla interactions, i.e., those in the form $\widetilde{H}=H\otimes I_a$ with $i H\in \mathfrak{v}_2$. This suggests that, even though $ \widetilde{V}=V\otimes I_a$, the presence of ancilla still  might be useful and indeed might be sufficient to achieve  universality.

To address this question we use the notion of  $l$-body phases. Let $\widetilde{\Phi}_l$ be the the $l$-body phase of unitary $\widetilde{V}=V\otimes I_a$. 
 From the the results of Sec.\ref{App:const}, we know that if   $\widetilde{V}=V\otimes I_a$ is realizable with $k$-local rotationally-invariant unitaries, then for all even integers  $l>k$,   $\widetilde{\Phi}_l=0 \ : \text{mod } 2\pi $. On the other hand, using proposition \ref{propnew}, we know that the $l$-body phase of 
$\widetilde{V}=V\otimes I_a$ is related to the $l$-body phase of $V$, via 
\be
\widetilde{\Phi}_l= 2\Phi_l \ \  \ \ \ : \text{mod } 2\pi\ . 
\ee
This means that for generic symmetric unitary $V$ all $l$-body phases of $\widetilde{V}$ are non-zero for $l=0,\cdots, 2\lfloor n/2\rfloor $, and therefore it cannot be realized using $k$-local symmetric unitaries with $k<2\lfloor n/2\rfloor $. In fact, 
comparing this with theorem \ref{Thm:Phi} that characterizes symmetric unitaries in terms of their $l$-body phases, we find a stronger result: Roughly speaking, adding the ancilla to a system with $n$ qubits does not increase the dimension of the Lie algebra of realizable unitaries on the $n$ qubits in the main system. More precisely,   the dimension of the Lie algebra associated to unitaries $\widetilde{V}=V\otimes I_a$, where $\widetilde{V}$ is realizable with $k$-local symmetric unitaries on $n+1$ qubits, is equal to the dimension of $\mathfrak{v}_k$, the Lie algebra associated to unitaries $V$ that are realizable with $k$-local symmetric unitaries on $n$ qubits.

 \end{document}